\documentclass[aps,pra,twocolumn,superscriptaddress,floatfix]{revtex4-1}
\usepackage{graphicx}
\usepackage{amsmath}
\usepackage{amsthm}
\usepackage{amssymb}

\theoremstyle{theorem}
\newtheorem{proposition}{Proposition}

\newtheorem{conjecture}{Conjecture}
\newtheorem{theorem}{Theorem}
\newtheorem{observation}{Observation}

\newcommand{\mean}[1]{\left\langle #1 \right\rangle}

\DeclareMathOperator{\per}{per}
\DeclareMathOperator{\sgn}{sgn}

\begin{document}

\title{Efficient computation of permanents, with applications to Boson
  sampling and random matrices}

\date{\today} 

\author{P. H. Lundow} 
\email{per.hakan.lundow@math.umu.se} 

\author{K. Markstr\"om}
\email{klas.markstrom@math.umu.se} 

\affiliation{Department of mathematics and mathematical statistics,
  Ume\aa{} University, SE-901 87 Ume\aa, Sweden}

\begin{abstract}
  In order to find the outcome probabilities of quantum mechanical
  systems like the optical networks underlying Boson sampling, it is
  necessary to be able to compute the permanents of unitary matrices,
  a computationally hard task.  Here we first discuss how to compute
  the permanent efficiently on a parallel computer, followed by
  algorithms which provide an exponential speed-up for sparse matrices
  and linear run times for matrices of limited bandwidth.  The
  parallel algorithm has been implemented in a freely available
  software package, also available in an efficient serial version.  As
  part of the timing runs for this package we set a new world record
  for the matrix order on which a permanent has been computed.

  Next we perform a simulation study of several conjectures regarding
  the distribution of the permanent for random matrices. Here we focus
  on permanent anti-concentration conjecture, which has been used to
  find the classical computational complexity of Boson sampling.  We
  find a good agreement with the basic versions of these conjectures,
  and based on our data we propose refined versions of some of them.
  For small systems we also find noticable deviations from a proposed
  strengthening of a bound for the number of photons in a Boson
  sampling system.
\end{abstract}

\keywords{Permanent, linear optics, boson sampling}

\maketitle

\section{Introduction}
One of the major tasks for experimental quantum physics today is to
implement and verify the performance of a universal quantum
computer. Under common complexity-theoretical assumptions such a
machine is expected to be able to solve certain problems far more
efficiently than any classical computer.  However, the experimental
task is formidable and recently an intermediate step known as Boson
Sampling \cite{aaronson:13} has become a focus for both theoretical
and experimental work.  The output probabilities of a Boson sampling
process are given by the \emph{permanent} of submatrices of the
unitary matrix describing the system.  The permanent of an $n\times
n$-matrix $A=(a_{i,j})$ is defined as
\begin{equation}\label{per}
  \per(A) = \sum_{\pi\in S_n}\prod_{i=1}^n a_{i,\pi(i)}
\end{equation}
where the sum is taken over all $n!$ permutations of
$N=\{1,2,\ldots,n\}$.  It differs from the determinant by ignoring the
sign of the permutation:
\begin{equation}\label{det}
  \det(A) = \sum_{\pi\in S_n}\sgn(\pi)\prod_{i=1}^n a_{i,\pi(i)}
\end{equation}
While the determinant $\det(A)$ can be computed in polynomial time,
computation of $\per(A)$ is in fact \#P-hard, even for
$01$-matrices~\cite{VAL}.  So, the expectation~\cite{aaronson:13} is
that Boson Sampling can be efficiently implemented as a physical
quantum system but computing the output probabilities will be hard for
a classical algorithm, while still not leading to a universal quantum
computer.

In order to verify that Boson sampling experiments give results which
agree with the theoretical prediction we need to compute output
probabilities for as large systems as possible.  This can be done
either exactly~\cite{Tian} or approximately through
simulation~\cite{Samp}.  Here the simulation algorithms still rely on
exact computation of some permanents as one step in the algorithm.
Our aim has been to implement an efficient, parallel and serial,
freely available, software package for computation of permanents,
which can either be used for direct computation of Boson sampling
output probabilities or to speed-up the existing sampling
programs. The resulting software package is freely
available~\cite{perm-code}.  On a parallel cluster a program based on
this package is able to compute the permanent of a $54\times 54$
matrix, where the previous record from \cite{Tian} was 48, in less
core time than the previous record.  Our package can be compiled for
both ordinary desktop machines and supercomputers.

The hardness argument in \cite{aaronson:13} is based on some
conjectures regarding the distribution of the permanent for certain
random matrices.  The behavior of random permanents has also become a
focus for some of the scepticism regarding quantum computing
\cite{2014arXiv1409.3093K}.  Using our package we have performed a
large-scale simulation study of the permanent distribution for several
different families of random matrices and the second part of our paper
is a discussion of this in relation to the conjectures from
\cite{aaronson:13}, as well the mathematical results from
\cite{MR2483225} and \cite{Grier2016NewHR}. In short, we find support
for some of the mentioned conjectures and propose some modifications
based on the sampling data.  We also find that one conjecture from
Ref.~\cite{aaronson:13}, which relates the number of Bosons to the
number of modes in a Boson sampling system, does not agree well with
data for small matrices.  This might change for larger sizes but
nonetheless means that eve more care must be taken in the analysis of
small Boson sampling experiments.


\section{Algorithms, the program library and its performance}
We will here discuss how to speed-up the computation of $\per(A)$,
both for general matrices and for matrices with some kind of
additional structure.  Following this we will look at the performance
of our implementation of these algorithms, both in terms of speed and
precision.


\subsection{A parallel algorithm for permanents of general matrices}
As it is formulated in Eq.~\eqref{per} it would take $n\cdot n!$ steps
to compute the permanent. However, it was shown by
Ryser~\cite{ryser:63} that it can be formulated as
\begin{equation}\label{ryser}
  \per(A)=(-1)^n\,\sum_{J \subseteq N} (-1)^{|J|} \prod_{i=1}^n
  \sum_{j \in J} a_{i,j}
\end{equation}
where $N=\{1,2,\ldots,n\}$. This reduces the number of operations to
about $n^2 2^n$. A further improvement is to use Gray-code ordering of
the sets $J$ in Eq.~\eqref{ryser} as noted in
Ref.~\cite{wilf:78}. Finding the next set in this order takes on
average $2$ steps.  Ref.~\cite{wilf:78} also shows how to halve the
number of steps by only summing over subsets of
$\{1,2,\ldots,n-1\}$. The number of operations then is reduced to the
currently best $n 2^n$. This improvement on Eq.~\eqref{ryser} means
that the improved method can compute a permanent for $n=50$ slightly
faster than the method in Eq.~\eqref{ryser} can for $n=45$.

In some applications, like Boson sampling, matrices may have repeated
rows or columns. With this in mind let us note that repeated columns,
or rows, allows us to speed up the calculation of the permanent,
without changing the basic form of Ryser's formula.  Assume that the
distinct columns of $A$ are $\bar{c}_1,\ldots,\bar{c}_R$ and that $A$
has $m_i$ columns equal to $\bar{c}_i$.  Let $\Omega$ denote the set
of all vectors $(f_1, f_2, \ldots, f_R)$ such that $0\leq f_i\leq
m_i$. Then
\begin{multline}\label{ryser2}
  \per(A) =\\ (-1)^n \sum_{\Omega}(-1)^{\sum\limits_{t=1}^R f_t}
  \left(\prod_{j=1} ^{R} {m_j \choose f_j} \right ) \prod_{i=1} ^{n}
  \sum_{k = 1} ^{R} f_k \bar{c}_{i}(k)
\end{multline}
Using the weighted form of Ryser's formula automatically leads to a
speed-up when repeated columns are present.  A very rough upper bound
on the number of terms in this sum is $n^R$. So for $R$ small compared
to $n$ one gets a significant speed-up compared to general matrices,
and for constant $R$ the algorithm runs in polynomial time. Using a
slightly more involved expansion this can be extended \cite{Ba96} to
an algorithm with a running time of the form $\mathcal{O}(n^{cr})$,
where $r$ is the rank of the matrix and $c$ is some constant..

In the Appendix we give explicit algorithms for computing the
permanent on a parallel computer.  Here we will give a short
discussion of the performance, both with respect to running time and
precision, of our program for computing the permanent function. In
Appendix \ref{alg} we give a more detailed description of the
algorithm, and instructions on how to download the program package.
 
Our benchmarking, and our later numerical simulations, were performed
on the Kebnekaise cluster at HPC2N in Ume\aa{}, Sweden.  Each node on
this cluster has two 14-core Intel Exon processors running at $3.5$
GHz, and 128 GB RAM.


\subsection{Precision}

We start by examining issues concerning precision.  The
implementations of different algorithms for computing the permanent in
\cite{Tian} achieved far worse precision for Ryser's algorithm than
for some other variations, leading to errors larger than 100\% for
$n\geq 32$. However, by a judicious choice of precision for single
numbers and summation method we find that Ryser's algorithm can be
implemented with both good precision and speed.

Note that the formulation in Eq.~\eqref{ryser} is essentially a very
long sum of products. The terms, of different signs, can vary greatly
in size. Some care must be taken to make sure that the resulting sum
is relevant. We have experimented with three approaches: doing all
computations with standard double precision, computing the product
with double precision (usually safe precision-wise) and using
quadruple precision for the sum, or, using only double precision and
Kahan summation~\cite{kahan:65} for the sum.

Using only double precision is of course the fastest but caution is
needed precision-wise if $n\gtrsim 30$. The double-quadruple precision
approach runs about half as fast but the precision is quite superior.
We see no significant difference in precision between Kahan summation
and the double-quadruple precision approach. Also, Kahan summation
only runs about 5\% slower than standard summation in double
precision.  Thus Kahan summation is highly recommended, especially if
quadruple precision is not available.  Some care must be taken so that
the compiler is not given too much freedom to alter the code
semantically during optimization. On the different compilers we tried
the default optimization setting did, however, not alter the code.
Note that partial sums from each core or node should be stored in an
array before the final Kahan summation, rather than just performing a
Reduce-operation. The double-quadruple precision approach, on the
other hand, allows for using the built-in Reduction-routines when
summing up the partial sums and the use of full optimization.

It is, of course, difficult to say in general what the true precision
of the result is after a permanent computation. Only for all-$1$
matrices, denoted $\mathbf{J}$ (recall $\per(\mathbf{J})=n!$), and
some $01$-matrices corresponding to biadjacency matrices of graphs, do
we have exactly known permanents. We suspect, however, that these are
worst-cases precision-wise, and that e.g., the random matrices used in
our later simulations are somewhat safer.  Comparing the results on
random Gaussian matrices using respectively double-precision and
double-quadruple precision shows much smaller errors than for
$\mathbf{J}$-matrices. The number of agreeing digits, for both cases,
is roughly $17-0.15n$ digits. For the range of $n$ studied here this
never puts the difference above $10^{-11}$ so computational error
seems not to be an issue in this study.

In Fig.~\ref{fig:error} we show the relative error in computing
$\per(\mathbf{J})$ using different computational scenarios. Using only
double precision is sensitive to how the computation is done. The
top-most set of points in the figure shows the error from computing
with double precision on a single core. The middle set of points show
the error when the computation (still only double precision) is run on
a single multi-core node ($20\le n \le 37$) and several multi-core
nodes ($38\le n \le 54$). It thus matters in which order the partial
sums are computed and added, this is here left to Reduction (OpenMP)
and AllReduce (MPI) which clearly do a good job. The bottom set of
points shows the error when the double-quadruple precision approach is
used. Here the order of summation now matters a lot less. The errors
for single-core ($15\le n\le 40$) and multi-core\&node ($20 \le n \le
53$) are here almost indistinguishable from each other so that
Reduction-AllReduce have no significant influence.

\begin{figure}
  \includegraphics[width=3.4in]{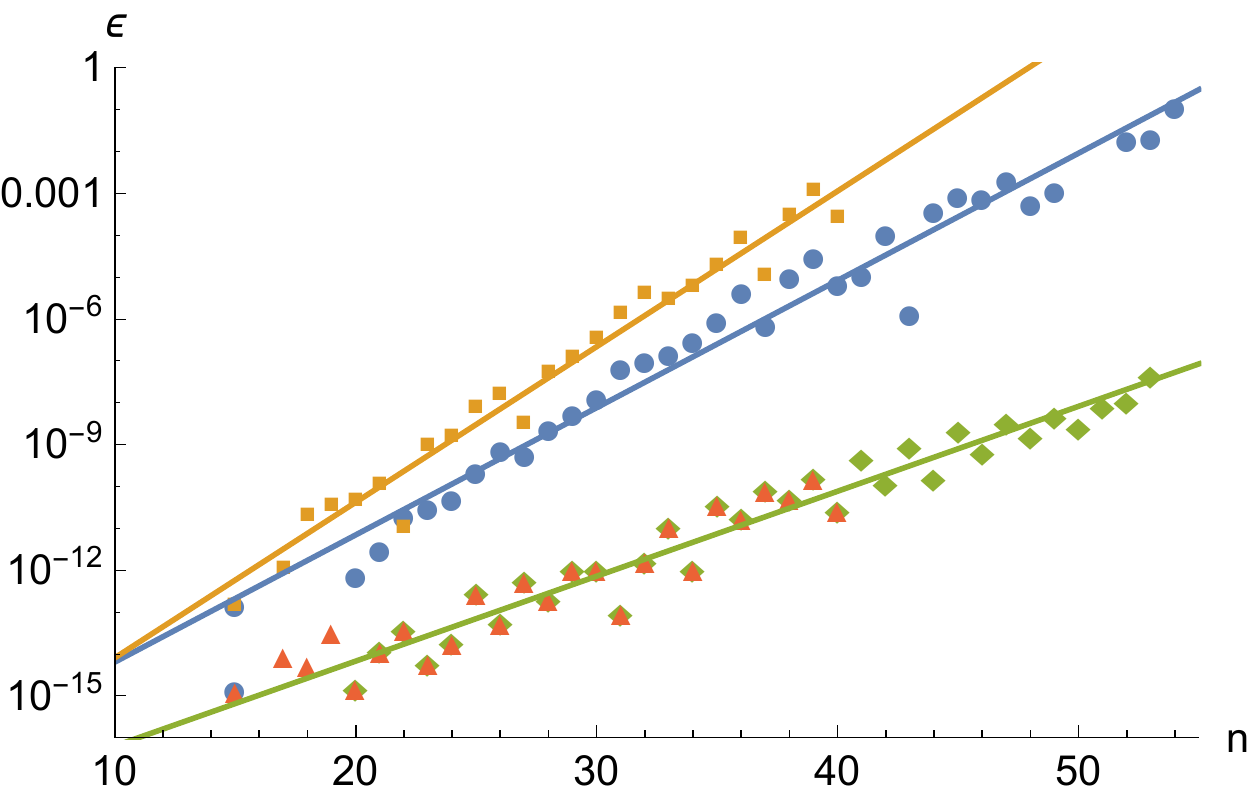}
  \caption{\label{fig:error} The relative error
    $\varepsilon=(\hat{p}-p)/p$ versus $n$ for double precision on
    single-core jobs (upper set of points), double precision on
    multi-core and multi-node jobs (middle set of points) and
    double-quadruple precision on single- and multi-core and
    multi-node jobs (lower set of points). Matrix is $\mathbf{J}$ (see
    text). From the fitted lines we obtain the decimal
    digit errors, respectively, $0.37n-17.8$, $0.30n-17.2$ and
    $0.20n-18$.}
\end{figure}

\subsection{Running time}
In Fig.~\ref{fig:time} we show the total core time for computing the
permanent of an $n\times n$-matrix using the double-quadruple
precision approach. The lower set of points ($3\le n \le 40$) are for
single-core and the upper set of points ($20\le n\le 53$) are for
multi-core single- and multi-node. For the smaller $n$ we obviously
had to repeat the calculation many times to get an estimate of the
rather small running times. The almost constant run-time between
$20\le n\le 30$) is due to the overhead time (ca $5$ seconds) for
starting the program with MPI. There is of course an overhead time for
OpenMP multi-core as well, but this is much smaller. Several nodes are
only used from $n=38$, ranging from $2$ up to $400$ nodes for
$n=51,52,53$. The black line indicates a core-time of roughly
$10^{-9.2} n 2^n$ seconds.

\begin{figure}
  \includegraphics[width=3.4in]{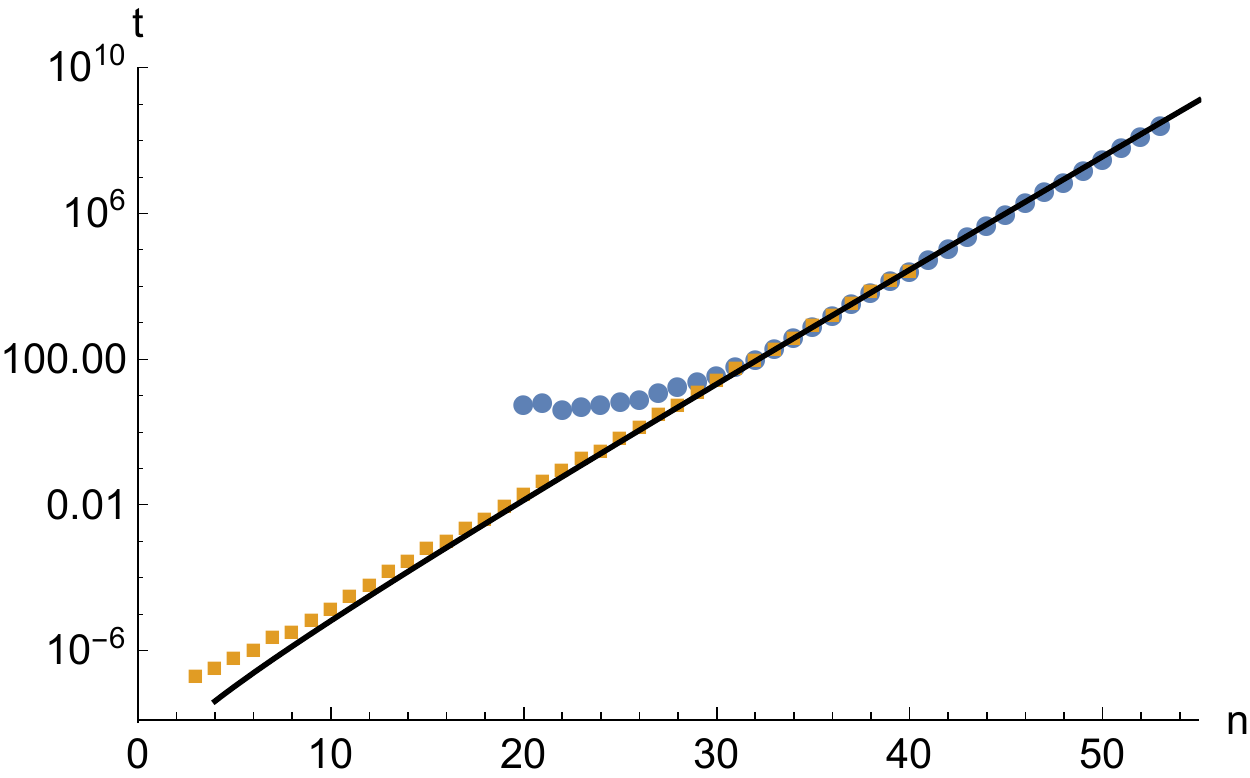}
  \caption{\label{fig:time} The total core time for a single-core job
    ($3\le n \le 40$, lower set of points, orange squares) and
    multi-core jobs on one or more nodes ($20\le n\le 53$, blue
    points) for a double-quadruple precision implementation of the
    permanent. For the multi-core jobs MPI was used for global
    communication between nodes, thus the overhead time of ca $5$
    seconds for $20\le n\le 30$. The line corresponds to an estimated
    run time of $10^{-9.2} n 2^n$ seconds.}
\end{figure}

The corresponding plot (not shown) for double-precision is very
similar but runs almost twice as fast.  The fluctuations between
different $n$ is however more pronounced but is not visible in a
log-plot. We should mention some effects from, we think, cache memory
sizes. For single-core the run time increases, as it should, by a
factor slightly larger than $2$. However, at $n=8, 16, 24, 32, 40$
this ratio is significantly less than $2$ between $n$ and $n-1$. For
example, at $n=8$ the ratio is $0.7$ so it is actually $30\%$ faster
to compute an $8\times 8$-permanent than a $7 \times 7$-permanent. At
$n=16, 24, 32, 40$ the ratio is respectively $1.0$, $1.2$, $1.4$ and
$1.6$.

\subsection{Improved algorithms for Sparse and Structured Matrices}\label{sec:iassm}
The basic algorithms for computing the permanent can be modified to
handle matrices with many zero entries in a more efficient way, and we
will here give a brief discussion of this.  As discussed in
Ref.~\cite{brod} low-depth Boson sampling set-ups lead to sparse
matrices, and there it is also proven that under standard complexity
theoretic assumptions computing these permanents is exponentially hard
even with a constant, but sufficiently large, number of non-zero
elements in each row and column. In Ref.~\cite{brod} the author
suggests that algorithms for this kind of matrix would be of interest
and here we demonstrate that they allow for an exponential speed-up
over the general case.

The first interesting class is general sparse matrices, i.e., matrices
where a significant proportion of the entries are zero.  For matrices
of this type many of the products in Ryser's method \eqref{ryser} will
be zero, and thus not necessary to compute.  If we interpret the
matrix $A$ as the adjacency matrix of an edge-weighted graph $G$ we
find that a set $J$ can only lead to a non-zero product if $J$ is a
dominating set in $G$, i.e., every vertex in $G$ has at least one
neighbor in $J$.  

If we restrict Eq.~\eqref{ryser} to sets $J$ which are dominating sets
we still have an exponential time algorithm, but running in time
poly$(n)a^n$, where $a$ can be noticeably smaller than 2.  Listing the
minimal dominating sets of $G$ is itself an exponentially hard
problem, which can be solved in time $\mathcal{O}(1.7159^n)$
\cite{Fomin}. However, for sufficiently sparse graphs we can instead
use some simple heuristics which lead us to include both all
dominating sets and some non-dominating ones, and still get a
significant speed-up compared to the basic version of Ryser's method.
In the sparse version of our code we have implemented this by
greedily, according to vertex degree, picking a set of vertices $I$
with disjoint neighborhoods and then listing all subsets $J$ which
contains at least one neighbor for every vertex in $I$.

For matrices leading to a sparse graph $G$ we can also prove that the
number of dominating sets is exponentially smaller than $2^n$, leading
to an exponential speed-up over the basic version of Ryser's formula.
We say that a matrix is \emph{$d$-sparse} if each row and column
contains at most $d$ non-zero entries.
\begin{theorem}\label{thm1}
	Let $A$ be a $d$-sparse $n \times n$ matrix. Then the permanent of
    $A$ can be computed in time
	\[\mathcal{O}(n2^n (1-2^{-d})^{n/d^2})\]
\end{theorem}
We will return to this theorem and give a proof in
Appendix~\ref{alg2}.

In Fig.~\ref{fig:sparse} we show how the running time increases with
$n$ in our Fortran implementation. Clearly this allows for computation
of the permanent of surprisingly large matrices.  However, with this
implementation the improvement quickly vanishes with densities above
$5$\%, a more ambitious algorithm would surely improve considerably on
this.

\begin{figure}
  \includegraphics[width=3.4in]{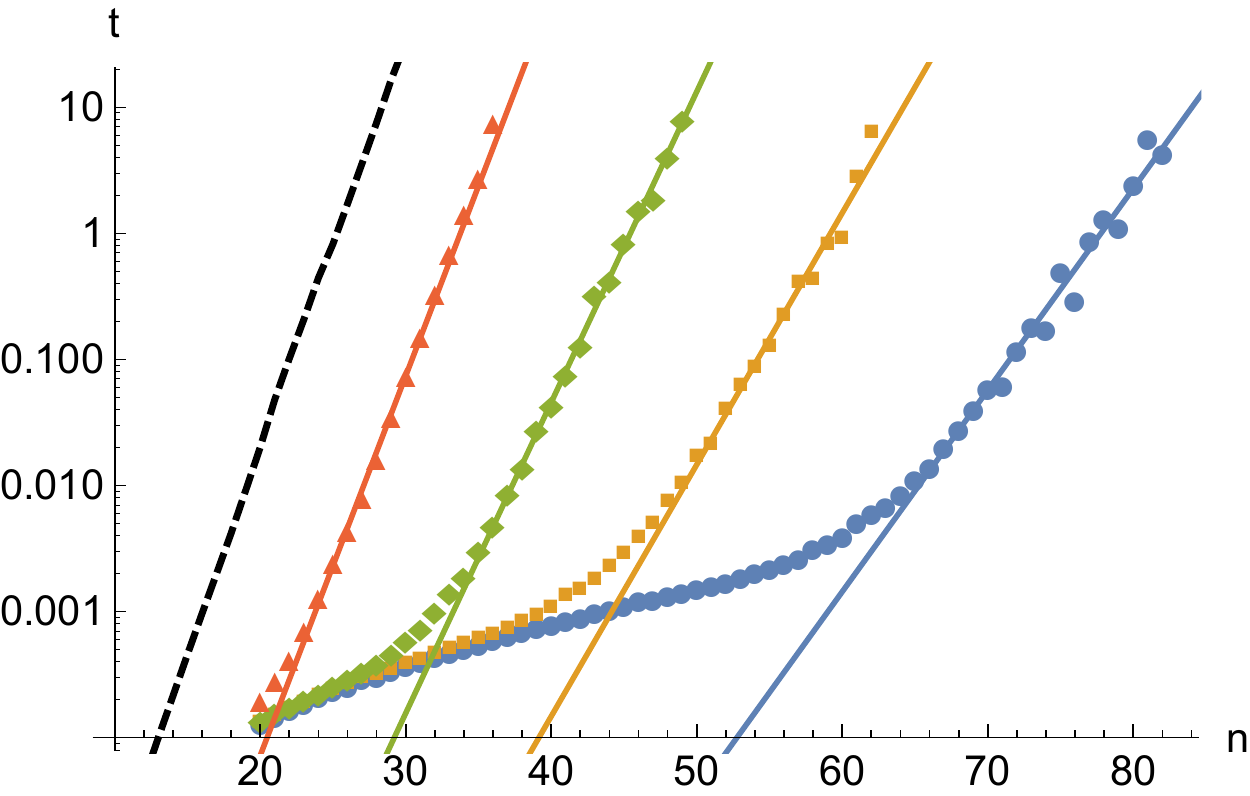}
  \caption{\label{fig:sparse} The median run time in seconds for the
    sparse matrix algorithm for real matrices. Matrices have $1$ on
    the diagonal (identity matrix) to guarantee a positive
    permanent. The off-diagonal entries are $1$ with probability $p$
    and $0$ with probability $1-p$. From left to right: standard
    algorithm for dense matrices (black dashed curve), $p=0.05$
    (orange triangles), $p=0.02$ (green diamonds), $p=0.01$ (yellow
    squares) and $p=0.005$ (blue points). The lines have slope
    $\ln(2)\approx 0.693$, $0.693$, $0.567$, $0.459$ and $0.368$ (left
    to right).}
\end{figure}

The associated graph $G$ can also be used to describe the second class
of matrices for which we have improved algorithms, namely those for
which $G$ has bounded \emph{tree-width}.  This direction is a
classical topic \cite{COURCELLE200123} and it is well known that for
matrices with tree-width bounded by some number $k$ the permanent can
be computed in time which is exponential in $k$ but polynomial in $n$.
After the first such algorithms a number of improvements have been
made \cite{Flarup2007,vanRooij2009,Meer2011,CIFUENTES201645}, but we
are not aware of any efficient implementations of the general bounded
tree-width methods.  However, one interesting subclass of this family
is the set of matrices of limited bandwidth $k$, i.e., matrices $A$
where $a_{i,j}=0$ whenever $|i-j|> k$, for which we can give a
practically useful algorithm.  This class of matrices is of particular
interest in connection with Boson sampling, since they include the
unitary matrices for Boson systems with certain restrictions.

For matrices of bandwidth $k$ we can interpret the product in
Eq.~\eqref{per} as a directed walk on the associated graph $G$ where
every vertex has out-degree one and in-degree one, and when the
bandwidth is limited to $k$ every vertex is connected to only vertices
which differ by at most $k$ from its own index $i$.  This means that
the sums over these weighted paths can be done using a transfer
matrix, following the general set up described in
Refs.~\cite{L2001,LM1,LM2}, leading to an algorithm which runs in time
$\mathcal{O}(n^2 2^k)$, and can in fact be reduced to linear time,
in $n$, for fixed $k$.  We describe a linear time version of this
method in Appendix~\ref{alg3}.

We have implemented the band-limited method in Mathematica and in
Fig.~\ref{fig:bw} we display the run times for $n\leq 100$ and
$k=1,2,3,4,5$ for random real gaussian matrices.  Our Mathematica code
is available online and this method will be included in an updated
Fortran package as well.

\begin{figure}
  \includegraphics[width=3.4in]{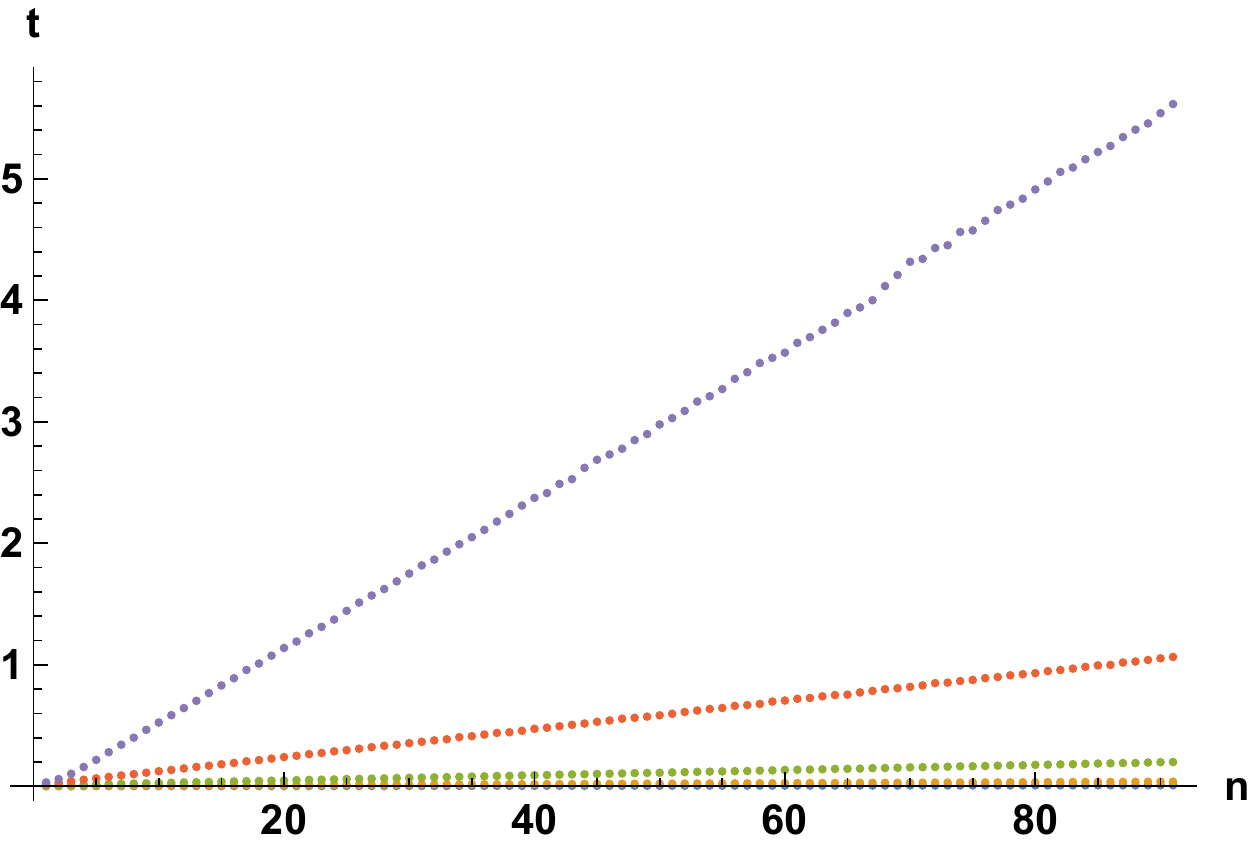}
  \caption{\label{fig:bw} The mean run time in seconds for the
    bandwith limited algorithm, for matrices with real Gaussian
    entries for bandwidth, $k=1,\ldots, 5$ (upwards in figure), and
    $n\leq 100$. }
\end{figure}

For matrices of logarithmic bandwidth this algorithm retains a
polynomial run time, but with a degree which depends on the scaling of
the logarithm. From the above we then get the following proposition.

\begin{proposition}
  For $n\times n$ matrices of bandwidth $k \leq c \log n$ the
  permanent can be computed in time $O(n^{2+c})$.
\end{proposition}

For a 2-dimensional Boson sampling system where each beam splitter
only interacts with its nearest neighbours, or more generally those
within some fixed distance $r$, the related unitary matrix will have
limited bandwidth. Here the bandwidth depends linearly on both $r$ and
the depth of the optical network.  As discussed in, e.g.,
Refs.~\cite{jozsa06,CC} these systems can be simulated in
sub-exponential time.  As long as the input and output states of such
a network do not have several photons in a given mode, so called
collision free states, the relevant matrix will be of limited
bandwidth and our linear time algorithm can be applied.  If there are
collisions in the output the matrix will have repeated columns and the
matrix will still have limited band width as long as the number of
collisions is bounded, but the number of collisions now add to the
bandwidth.  However, note also that if we only have repeated
columns our transfer matrix based algorithm can be modified to handle
this generalized version of band-limited matrices. Here the "width"
will then instead correspond to the maximum number of non-zero
elements in a column.

\begin{observation}
	Using the algorithm for band-limited matrices as a subroutine the
    algorithm from Ref.~\cite{Samp} simulates band-limited Boson
    sampling systems in polynomial time as long as the number of
    collisions in the output state is bounded.
\end{observation}

For a Haar random unitary the caveat on bounded collisions is not
required, due to the Bosonic birthday paradox \cite{AK12,bunch}, which
guarantees that collisions will be unlikely if the number of Bosons is
not too large.  For low-depth systems this theorem does not apply, but
due to the bounded range of the interactions we still do not expect
large numbers of collisions in a single output mode for a random input
state. Providing an exact quantitative form of this statement,
converging to the Bosonic birthday paradox as depth increases, would
be of interest.


\section{Complex Gaussian matrices}
Let $A=(a_{i,j})$ be an $n \times n$-matrix of i.i.d.~complex Gaussian
numbers, i.e., $A$ is a member of the Ginibre ensemble
$\mathcal{G}(n)$. Each element $a_{i,j}$ can be generated by first
producing two independent real Gaussian numbers $g_1$ and $g_2$, using
the Box-Muller method~\cite{box:58}, and then $(g_1 + \imath
g_2)/\sqrt{2}$ is a random complex Gaussian with mean 0 and variance
$1$. Clearly, the expected value of the permanent of a matrix with
such entries is $0$, due to symmetry.  However, we are mainly
interested in the properties of the modulus $|\per(A)|$.  It is
known~\cite{aaronson:13} that $\mean{|\per(A)|^2} = n!$, and this is
true as long as the entries are independent real or complex numbers
with variance 1, so that it is more natural to work with the random
variable $X=|\per(A)|/\sqrt{n!}$ and thus have $\mean{X^2}=1$.

Remarkably, it is also shown in Ref.~\cite{aaronson:13} that
$\mean{X^4}=n+1$. In general, the authors of Ref.~\cite{aaronson:13}
can express the even moments exactly in terms of the expected number
of decompositions of a $k$-regular multigraph into disjoint perfect
matchings. An asymptotic result, following from the proof of van der
Waerden's conjecture on the permanent of doubly-stochastic matrices,
is then that $\mean{X^{2k}}\sim (k/e)^n$ for $k\ge 3$. 

We will denote the distribution function by $F(x,n) = \Pr(X\le x)$
where $X$ are samples from $n\times n$-matrices. The density function
is denoted $f(x,n) = F'(x,n)$. Also let $f(x)=f(x,\infty)$ and
$F(x)=F(x,\infty)$ but we also use $f(x)$ and $F(x)$ as generic forms
when the context is clear. The $x$-values are chosen in a geometric
progression of $16$ per decade (an interval of the form $\lbrack
10^{k-1},10^k)$ for some integer $k$). The density $f(x,n)$ is then
obtained from the difference quotient
$f(x_i)=(F(x_{i+1})-F(x_{i-1}))/(x_{i+1}-x_{i-1})$. The error in
$F(x)$ is estimated as $\epsilon(x)=\sqrt{F(x)\,(1-F(x)}/\sqrt{m}$
where $m$ is the number of samples. The error in $f(x)$ can then be
obtained in the usual fashion.

\subsection{Measuring up the distribution}
For each size $n$ we have computed the permanent of $10^6$ random
complex Gaussian $n\times n$-matrices, for $n=1,2,\ldots,30$, which we
will use in this section, and $10^5$ matrices for $n=31,\ldots,35$, to
be used later. Let us take a look at the moments $\mean{X^k}$ of the
sampled data. In Fig.~\ref{fig:CGmom1} we show the mean $\mean{X}$ vs
$1/n$. A fitted line suggests the limit mean value $0.684(1)$ where
the error estimate is obtained from fitting on the points $n\ge k$
with $k=6,\ldots,12$. The error bars in the figure were estimated from
bootstrap resampling. For the second moment, shown in
Fig.~\ref{fig:CGmom2}, we know that it should be $1$ and indeed there
is only noise-like deviation from the line $y=1$.

\begin{figure}
  \includegraphics[width=3.4in]{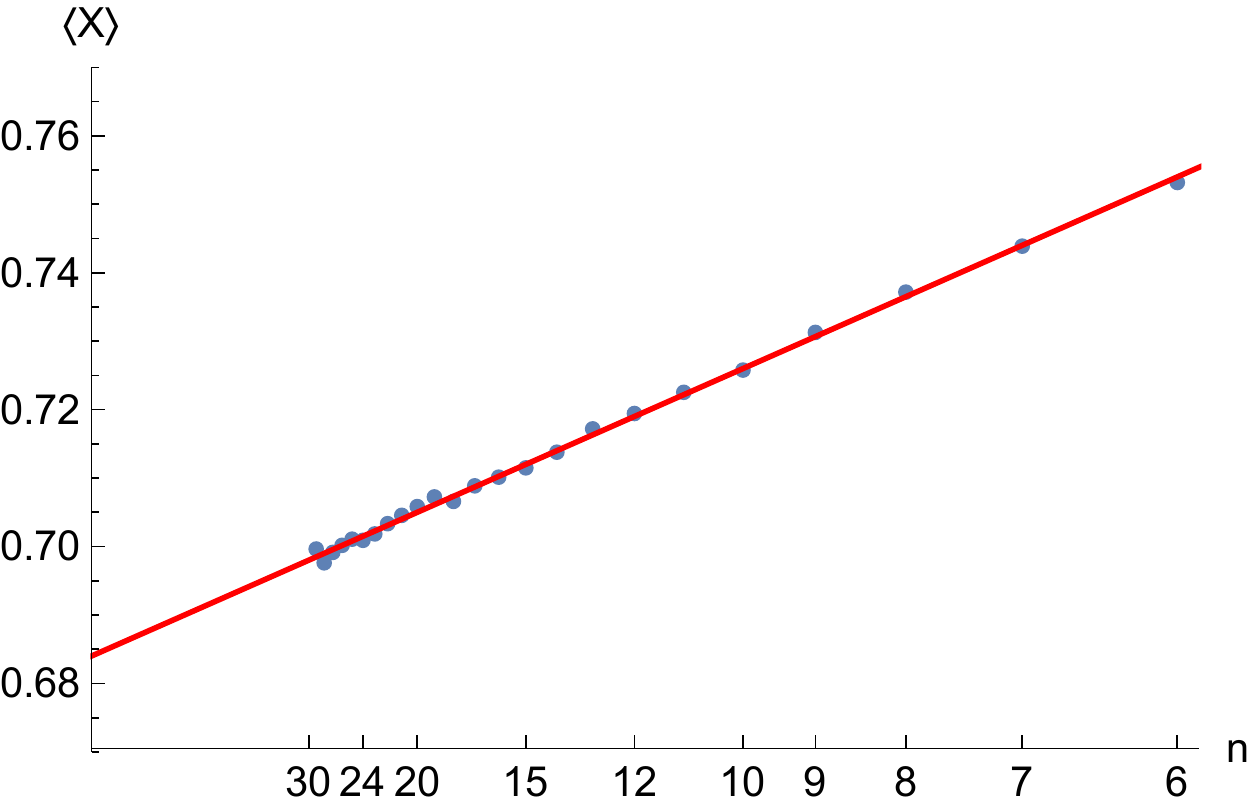}
  \caption{\label{fig:CGmom1} Gaussian entries. Mean value $\mean{X}$
    versus $1/n$ for $n=6,7,\ldots,30$ and the fitted line $y=0.684 +
    0.42x$ (red) where $x=1/n$.}
\end{figure}

\begin{figure}
  \includegraphics[width=3.4in]{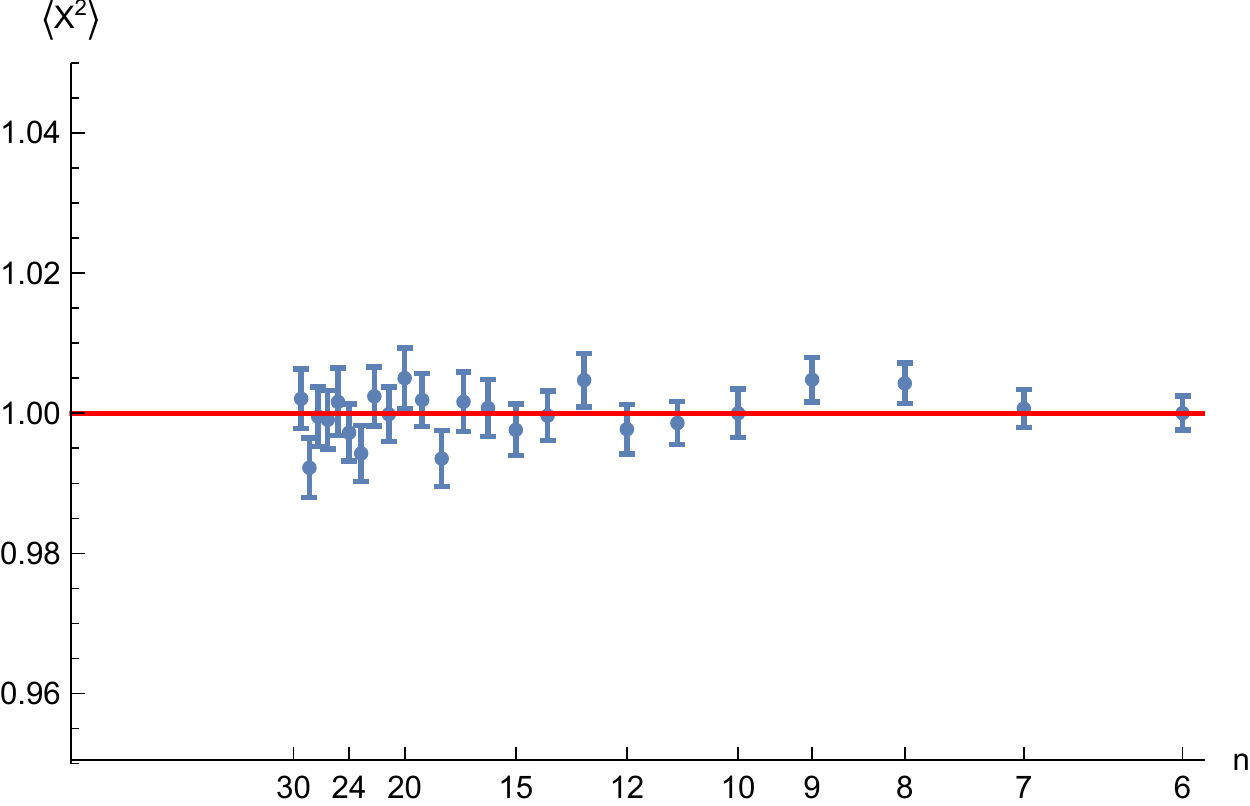}
  \caption{\label{fig:CGmom2} Gaussian entries. Mean value
    $\mean{X^2}$ versus $1/n$ for $n=6,7,\ldots,30$ and the line
    $y=1$.}
\end{figure}

In Fig.~\ref{fig:CGmom3} the third moment is shown and the noise is
still manageable. Using the line fitting procedure described for the
first moment, suggests the limit $3.20(3)$.  In Fig.~\ref{fig:CGmom4}
we show the fourth moment $\mean{X^4}$, which should be $n+1$, plotted
against $n$. A line $y=1+x$ together with the data points indicate
that for larger $n$ the data often favor a smaller moment. However,
the error bars occasionally get quite pronounced. This indicates that
our data set, from lack of samples, has not yet fully explored the
rather heavy tails of the distribution.

\begin{figure}
  \includegraphics[width=3.4in]{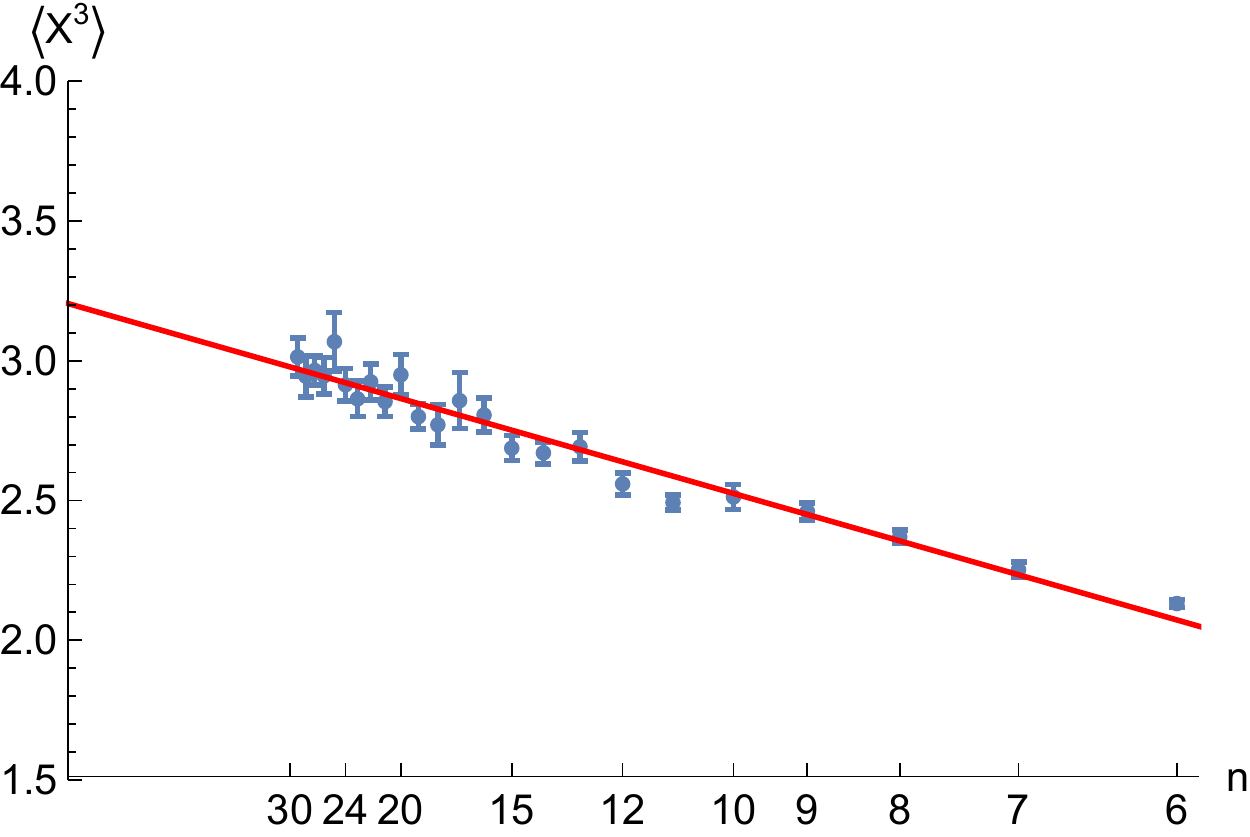}
  \caption{\label{fig:CGmom3} Gaussian entries. Mean value
    $\mean{X^3}$ versus $1/n$ for $n=6,7,\ldots,30$ and the line
    $3.20-6.8x$ (red) where $x=1/n$.}
\end{figure}

\begin{figure}
  \includegraphics[width=3.4in]{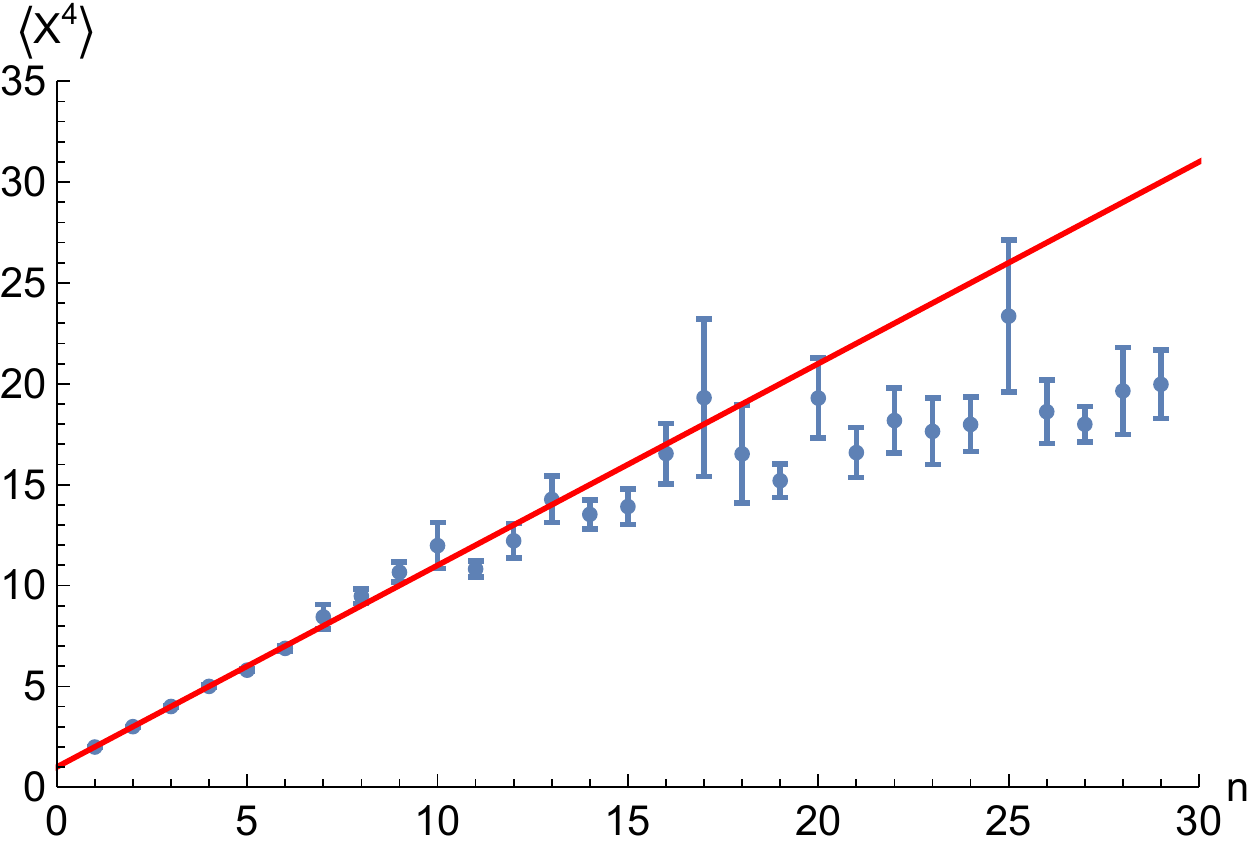}
  \caption{\label{fig:CGmom4} Gaussian entries. Mean value
    $\mean{X^4}$ versus $n$ for $n=1,2,\ldots,30$ and the line
    $y=1+n$ (red), i.e.~the correct expected value.}
\end{figure}

Finally, in Fig.~\ref{fig:CGdist1-n30} we show the distribution
density for $n=30$. This brings us to an interesting conjecture; the
permanent anti-concentration conjecture, see
Ref.~\cite{aaronson:13}. It states that there exists a polynomial $p$
such that
\begin{equation}
  F(1/p(n,1/\delta)) < \delta
\end{equation}
for all $n$ and $\delta>0$. Equivalently, this can be formulated as
the existence of constants $a$, $b$, $c$ such that
\begin{equation}
  F(\varepsilon) < c n^a \varepsilon^b
\end{equation}
for all $n$ and $\varepsilon>0$. In the next section we will provide
numerical support for this conjecture.

\begin{figure}
  \includegraphics[width=3.4in]{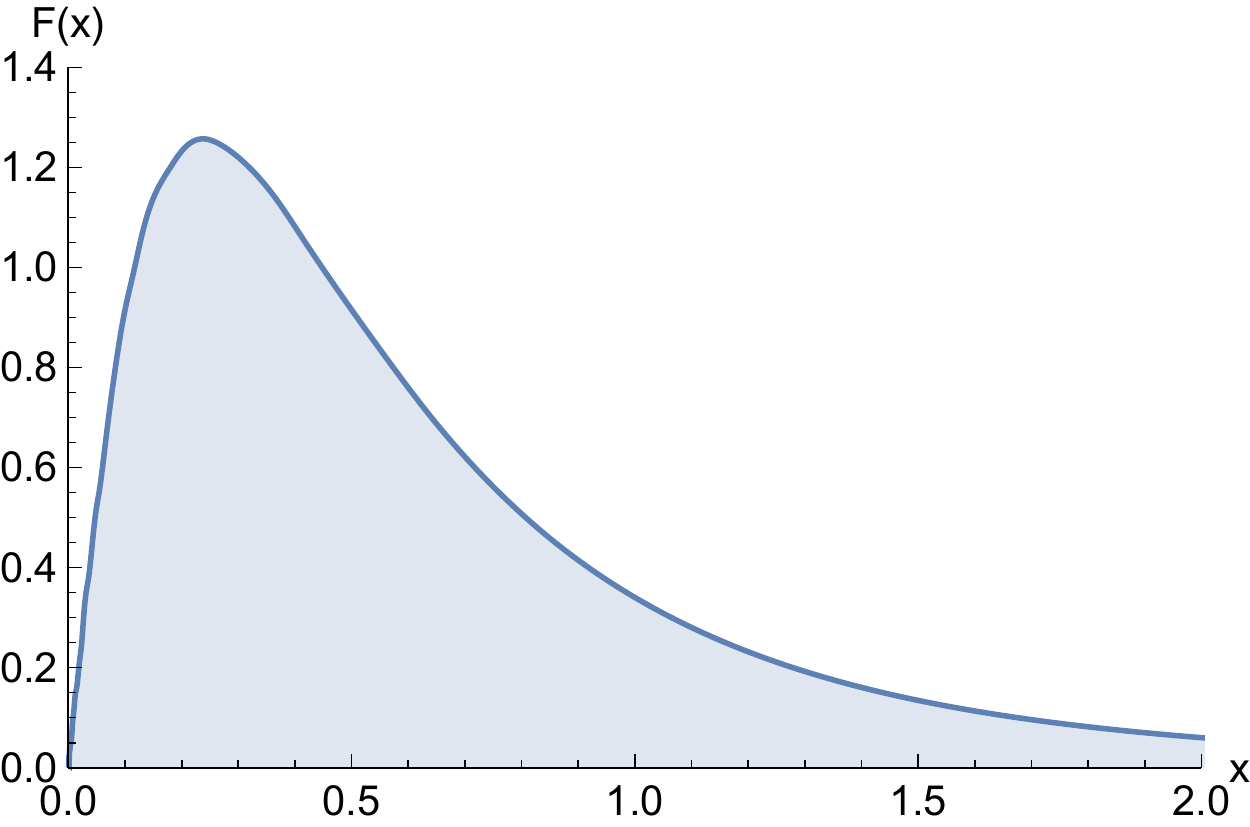}
  \caption{\label{fig:CGdist1-n30} Gaussian entries. Density function
    $f(x)$ versus $x$ for $n=30$.  Data from $10^6$ random $30\times
    30$-matrices $A$. This range ($x<2$) contains $95$\% of the
    samples.}
\end{figure}

\subsection{The distribution in detail}
Let us try to discern how $F(x)$ behaves for small $x$.  In
Fig.~\ref{fig:CGFxvn} we show $F(x,n)$ for $n=6,\ldots,30$ for a few
values of $x$ together with fitted lines. It turns out that $F(x,n)$
depends quite linearly on $1/n$ for $n\ge 6$. When smaller $n$ are
included a higher-order correction term becomes necessary. Note that
$F(x,n)$ is strictly increasing with $n$ for $x\lesssim 1.77$. For
each fixed $x$ we fit a line through the points $(1/n,F(x,n))$ and its
constant term then gives us an estimated asymptotic value, i.e.~we
make the ansatz $F(x,n)=F(x)+C(x)/n$ so that the slope $C(x)$ only
depends on $x$.  By deleting individual points from the line fitting
we obtain error estimates of the parameters of each line.  At
$x\lesssim 0.005$ the probabilities become less than $10^{-4}$ and the
error bar for each estimate of $F(x)$ becomes significant. We have
thus only used $x\ge 0.005$.

\begin{figure}
  \includegraphics[width=3.4in]{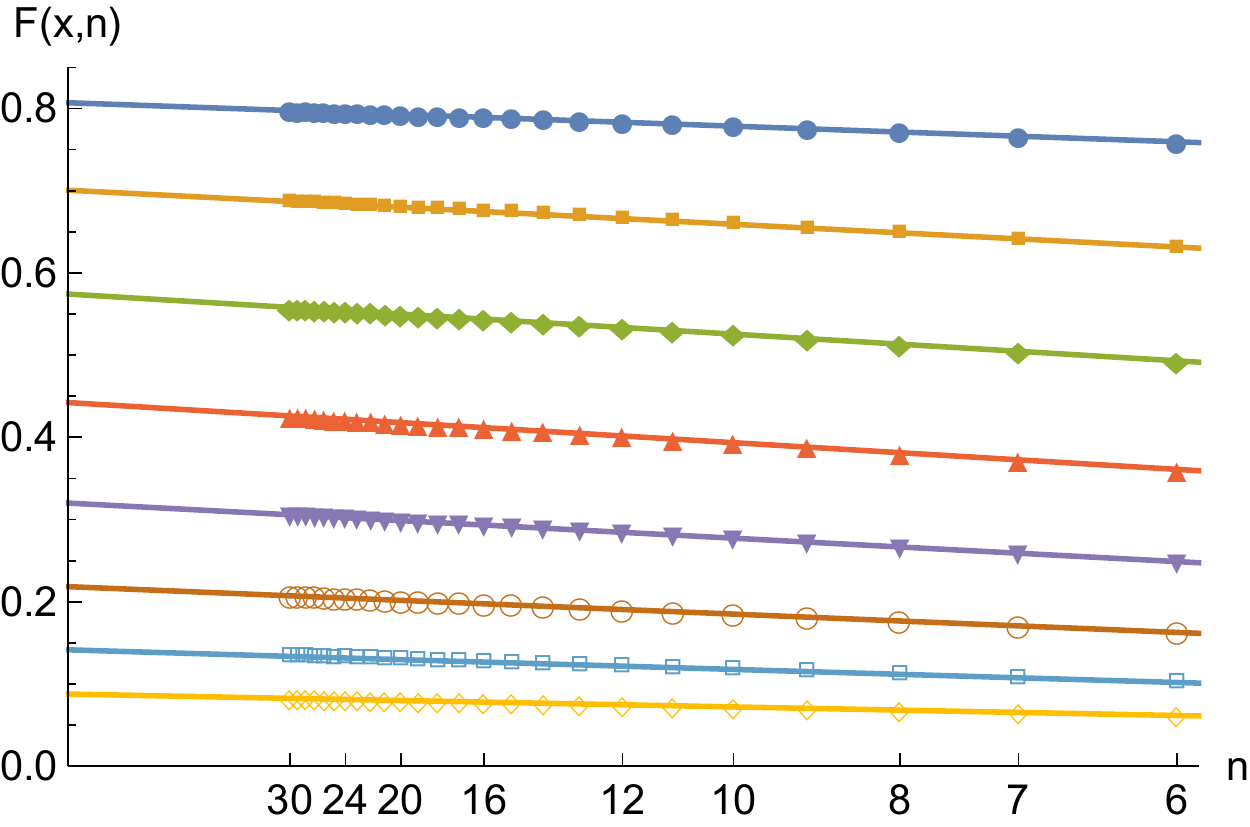}
  \caption{\label{fig:CGFxvn} Gaussian entries. $F(x,n)$ for
    $n=6,\ldots,30$ at $x=1$, $0.75$, $0.56$, $0.42$, $0.32$, $0.24$,
    $0.18$ and $0.13$ (downwards).  Lines are fitted to each set of
    points. Error bars are smaller than the points.}
\end{figure}

In Fig.~\ref{fig:CGlogFx} we show a log-log plot of the asymptotic
$F(x)$.  We have fitted a line with slope $2$ through the points
$x<0.05$ which corresponds to the approximation $F(x) \sim 6.0(1)
x^2$. The lower inset of the figure shows the ratio $F(x)/x^2$ which
plausibly approaches the limit $6$ despite some significant noise
beginning for $x\lesssim 0.01$.  The upper inset of
Fig.~\ref{fig:CGlogFx} shows a log-log plot of the difference
$6x^2-F(x)$ approximated by $7(1) x^3$ (note the rather wide error
bar), as indicated by the red line having slope $3$. Together they
suggest $F(x) \sim 6x^2-7x^3$. The error estimates captures how
sensitive the coefficients are to deleting data for individual $n$ and
which interval of $x$ we fit lines to.

There is also the matter of finite-size scaling to take into account,
i.e., the $n$-dependence. A similar analysis of the slopes of the
lines in Fig.~\ref{fig:CGFxvn}, i.e., the parameter $C(x)$ mentioned
above, gives that $C(x)\sim -12.0(5) x^2$.  Including the finite-size
term we then have the finite-size scaling
\begin{eqnarray}
  F(x,n) &\sim& \left(6 - \frac{12}{n}\right) x^2 - 7x^3,\label{CGFasym} \\
  f(x,n) &\sim& \left(12 - \frac{24}{n}\right) x - 21 x^2\label{CGfasym}
\end{eqnarray}
In the next subsection we will try to obtain a finite-size correction
for the $x^3$-term. In any case, we have so far found that
$F(x)\lesssim 6x^2$ for all $n$ and $x$ which supports that the
anti-concentration conjecture is true.

\begin{figure}
  \includegraphics[width=3.4in]{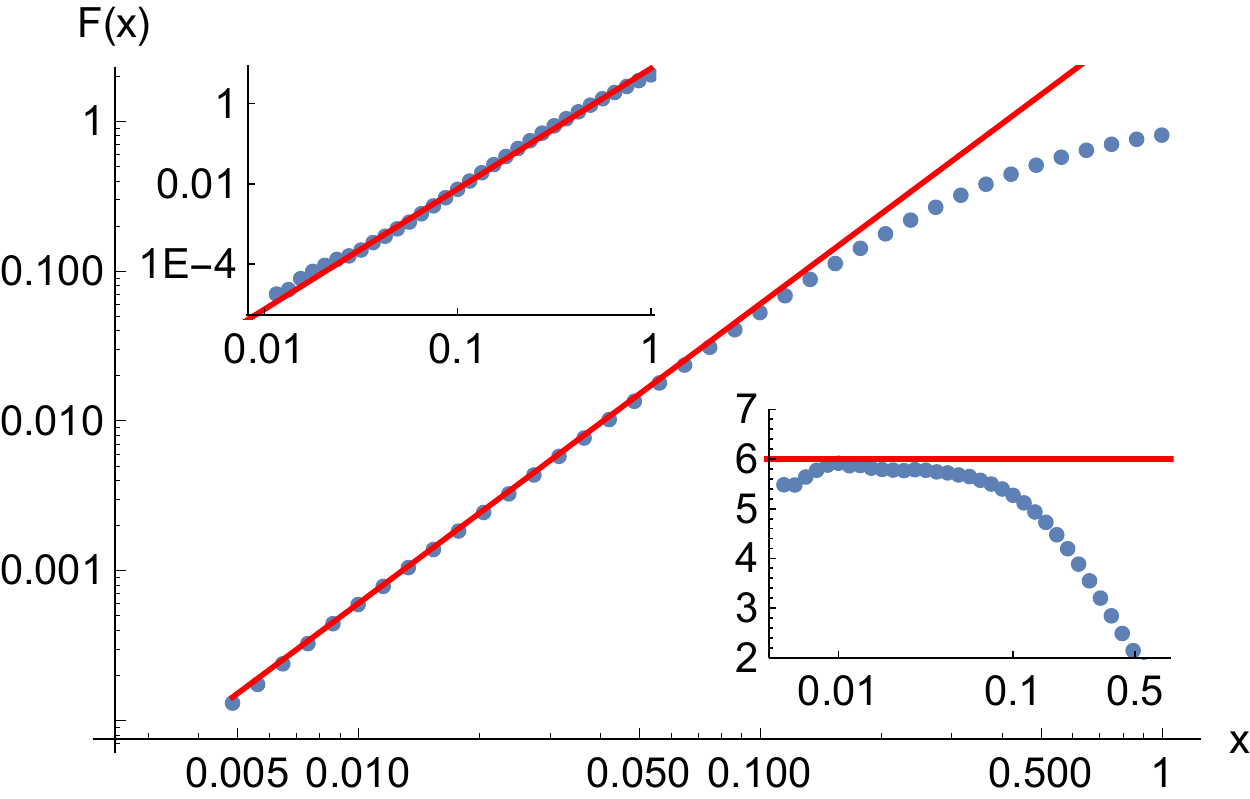}
  \caption{\label{fig:CGlogFx} Gaussian entries. Log-log plot of
    $F(x)$ and the line (red) corresponding to the approximation
    $6x^2$ (see text). The upper inset shows a log-log plot of $6x^2 -
    F(x)$ and a line (red) corresponding to $7x^3$. The lower inset
    shows the ratio $F(x)/x^2$ and the line (red) is at $y=6$.}
\end{figure}

\subsection{The distribution of the squares}
Let us now see how this translates to the distribution of the squares
$X^2$.  The distribution function of $X^2$ is $F_2(x)=\Pr(X^2\le x) =
\Pr(X\le \sqrt{x}) = F(\sqrt{x})$. This gives the density function
$f_2(x) = F'_2(x) = f(\sqrt{x})/(2\sqrt{x})$. Using Eq.~\eqref{CGFasym}
and Eq.~\eqref{CGfasym} we obtain
\begin{eqnarray}
  F_2(x) &\sim& \left(6 - \frac{12}{n}\right) x - 7x^{3/2}\label{CGF2asym} \\
  f_2(x) &\sim& 6 - \frac{12}{n} - 10.5\sqrt{x}\label{CGf2asym}
\end{eqnarray}

In Fig.~\ref{fig:CGdist2-n30} we show the density function $f_2(x,n)$
of $X^2$ for $n=30$. It was speculated in Ref.~\cite{aaronson:13} that
the density $f_2(x)$ goes to infinity when $x\to 0$ but we claim that
it goes to a limit value. Our approximation for $F(x)$ stays relevant
for $x\lesssim 0.10$ which from the perspective of the squares $X^2$
means that our formula for $f_2(x)$ is relevant only for $x<0.01$. We
are here at the lower $5$\% of our data set so our analysis demands a
large number of samples.

\begin{figure}
  \includegraphics[width=3.4in]{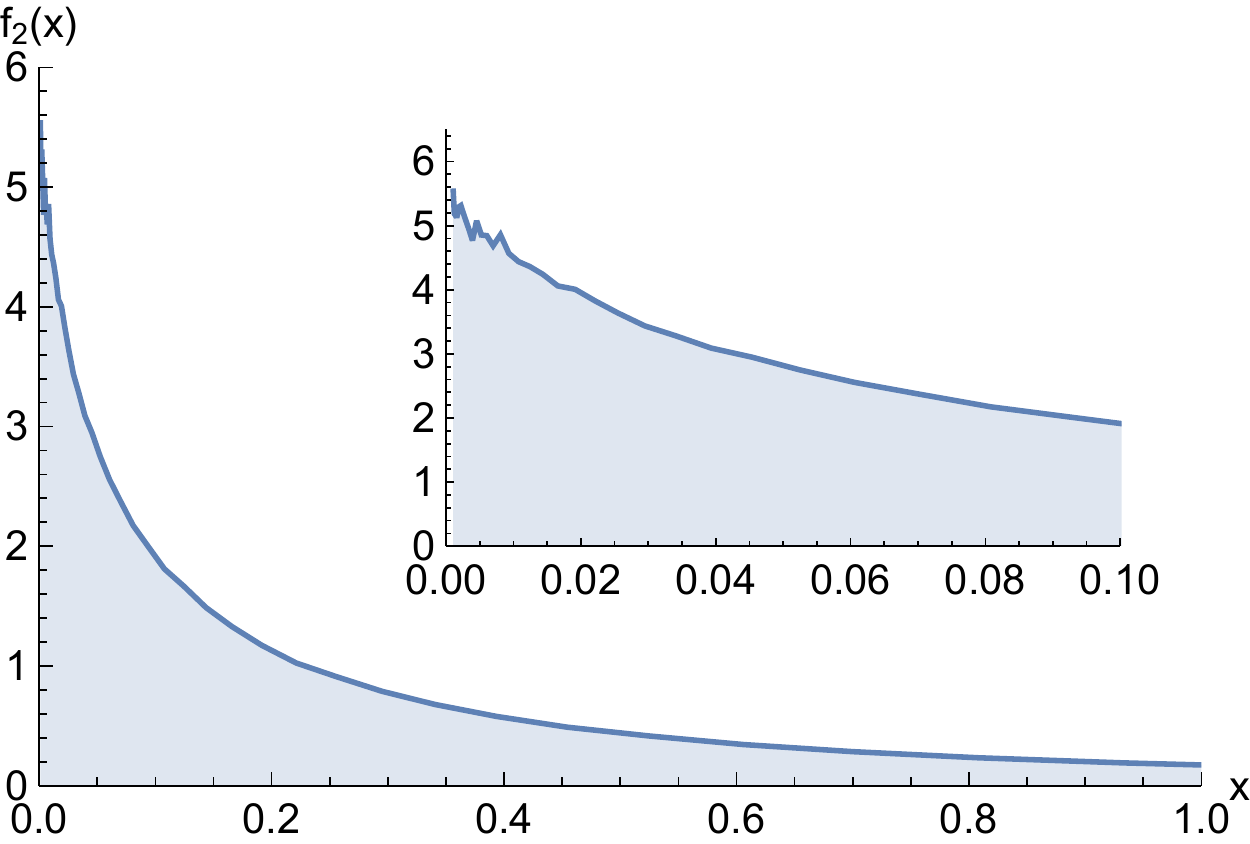}
  \caption{\label{fig:CGdist2-n30} Gaussian entries. Distribution
    density $f_2(x,30)$ of $X^2$ for $n=30$. The shown range covers
    $80\%$ of the samples. The inset shows a zoomed-in version
    convering $30\%$ of the samples.}
\end{figure}

Fitting the curve $y=C_0+C_1\sqrt{x}$ to the measured $f_2(x,n)$ for
the range $0.0005<x<0.01$ we find how the coefficients $C_0$ and $C_1$
depend on $n$. Note that $C_0$ should scale as $C_0=6.0(1)-12.0(5)/n$
as in Eq.~\eqref{CGf2asym}. This is confirmed in
Fig.~\ref{fig:CGf2scaling}.  The inset shows how $C_1$ depends on $n$
and we find $C_1=10.5(9) + 36(6)/n$ though the points are here quite
scattered, which is reflected in the error bars.  Again, the error
bars indicate how much the result depends on the choice of fitted
points (less) and range (more).

\begin{figure}
  \includegraphics[width=3.4in]{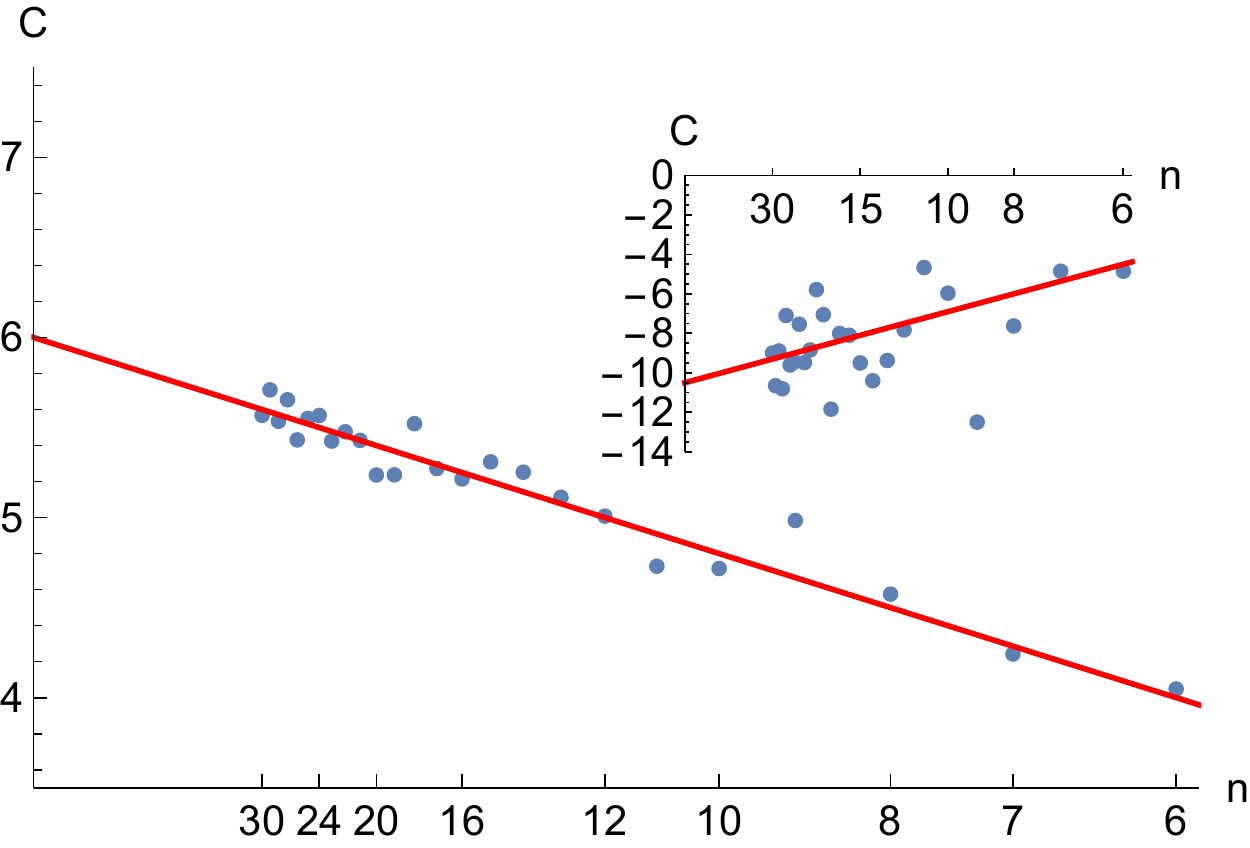}
  \caption{\label{fig:CGf2scaling} Gaussian entries. Scaling of the
    coefficients $C_0$ and $C_1$ (inset) for $f_2(x,n) = C_0 +
    C_1\sqrt{x}$ fitted to $0.0005<x<0.01$. The red lines are
    $y=6-12/n$ and $y=-10.5+36/n$ (inset).}
\end{figure}

Adding these new terms we find the following finite-size scaling
rules
\begin{eqnarray}
  F(x,n) &\sim& \left(6 - \frac{12}{n}\right) x^2 + \left(-7 +
  \frac{24}{n}\right) x^3,
  \label{CGFasym2} \\
  f(x,n) &\sim& \left(12 - \frac{24}{n}\right) x + \left(-21
  +\frac{72}{n}\right) x^2,
  \label{CGfasym2} \\
  F_2(x,n) &\sim& \left(6 - \frac{12}{n}\right) x + \left(-7 +
  \frac{36}{n}\right) x^{3/2},
  \label{CGF2asym2} \\
  f_2(x,n) &\sim& 6 - \frac{12}{n} + \left(-10.5 +
  \frac{36}{n}\right)\sqrt{x}
  \label{CGf2asym2}
\end{eqnarray}
with the limits
\begin{eqnarray}
  F(x) &\sim& 6 x^2 -7 x^3 \label{CGFlimit2} \\
  f(x) &\sim& 12 x - 21 x^2 \label{CGflimit2} \\
  F_2(x) &\sim& 6 x -7 x^{3/2} \label{CGF2limit2} \\ 
  f_2(x) &\sim& 6 -10.5 \sqrt{x}\label{CGf2limit2}
\end{eqnarray}

In Fig.~\ref{fig:CGfx} we show the measured $f(x,n)$ and the function
in Eq.~\ref{CGfasym2} for $n=10,20,30$ on $x<0.1$. The fit is quite
excellent. In Fig.~\ref{fig:CGf2x} we show the measured $f_2(x,n)$ and
Eq.~\eqref{CGf2asym2} for $x<0.008$, again with a good fit though the
error bars for $f_2$ are here quite noticeable.

\begin{figure}
  \includegraphics[width=3.4in]{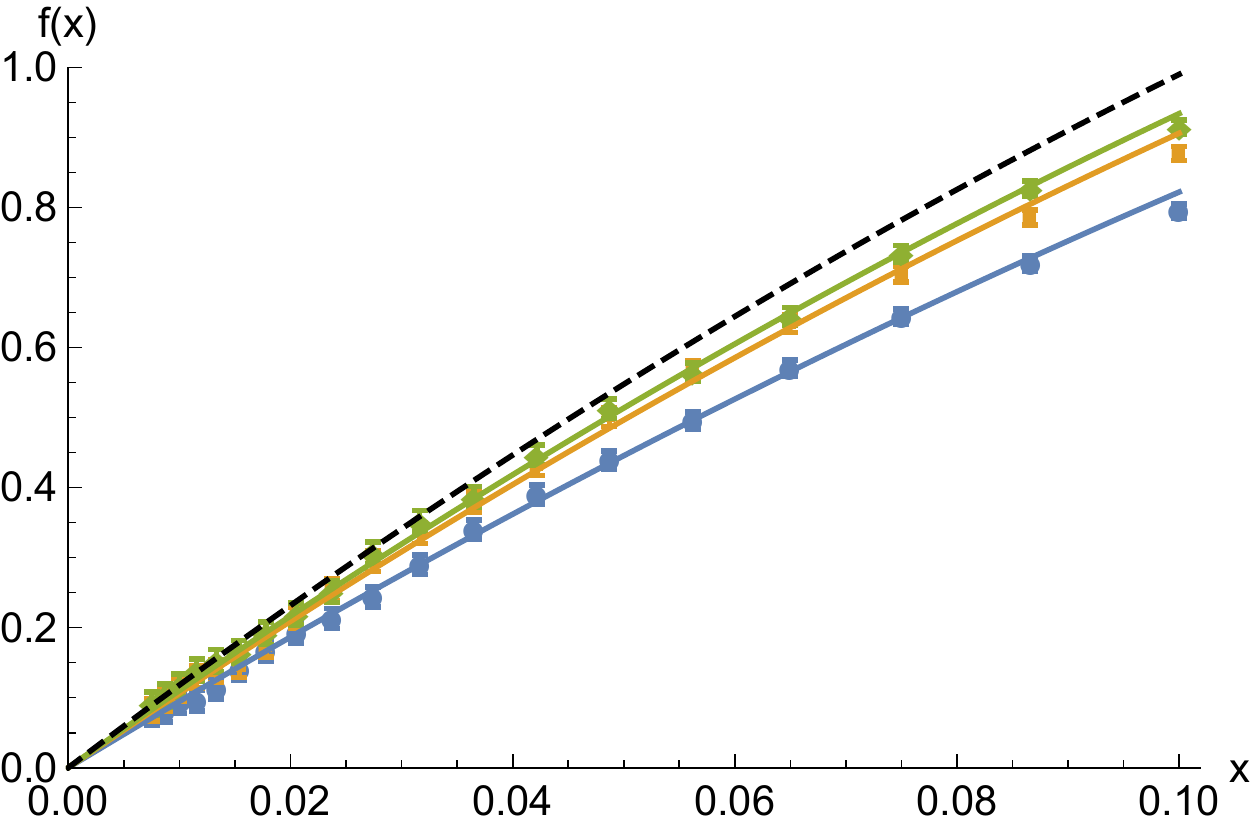}
  \caption{\label{fig:CGfx} Gaussian entries. Measured density
    $f(x,n)$ (points) and the scaling rule of Eq.~\eqref{CGfasym2} for
    $n=10, 20, 30$ (blue, orange, green curves; upwards) and $\infty$
    (dashed black curve). Error bars are smaller than the points.}
\end{figure}

\begin{figure}
  \includegraphics[width=3.4in]{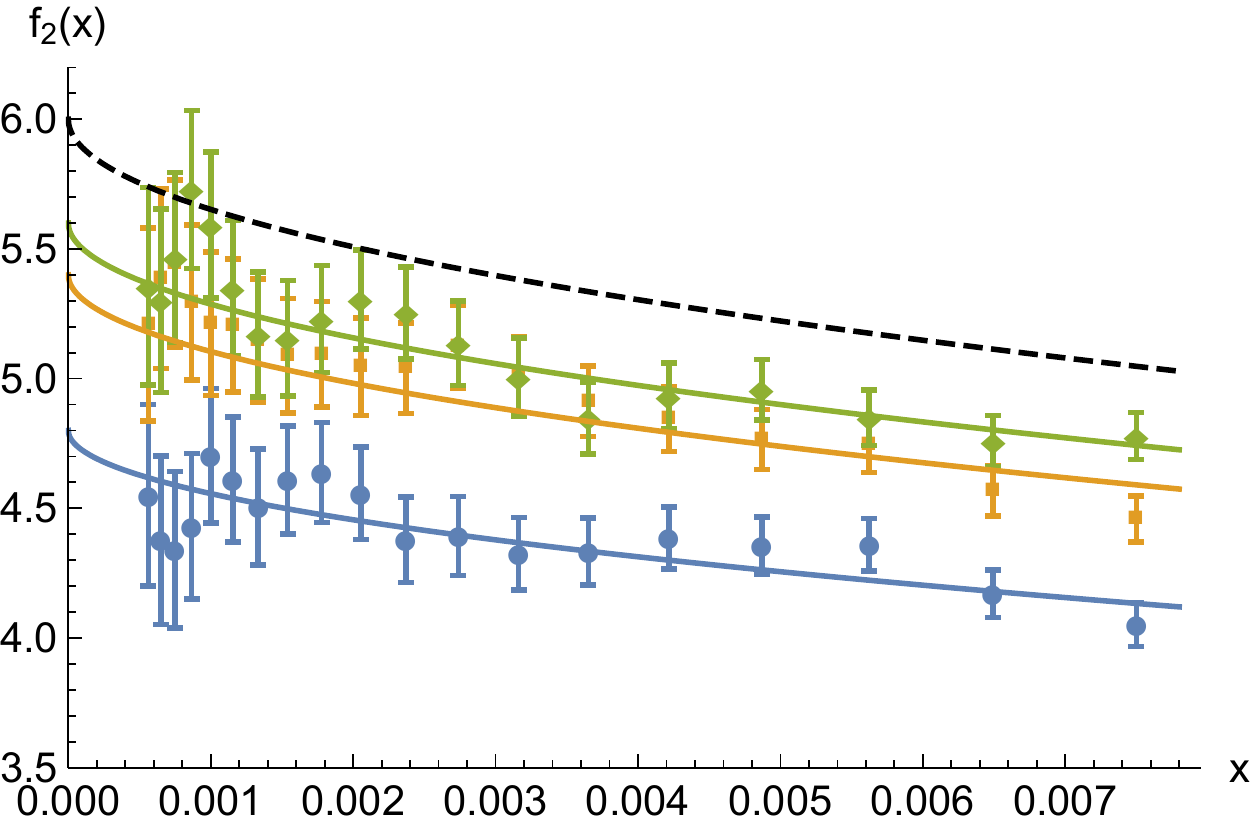}
  \caption{\label{fig:CGf2x} Gaussian entries. Measured density
    $f_2(x,n)$ (points) and the scaling rule of Eq.~\eqref{CGf2asym2}
    for $n=10$, $20$, $30$ (blue, orange, green curves; upwards) and
    $\infty$ (dashed black curve). Error bars become significant for
    $x\lesssim 0.001$.}
\end{figure}


\section{Complex circular distribution}
Here we let the entries of the matrix be random complex numbers of
modulus $1$, i.e.~we let each entry be of the form
$\exp(\imath\theta)$ where $\theta$ is uniformly distributed on the
interval $\lbrack 0,2\pi)$. For each size $n$ we have collected data
for $10^7$ such random matrices, for $n=1,2,\ldots,30$. As before, we
study the distribution of $X=|\per{A}|/\sqrt{n!}$. Our investigation
proceeds as in the previous section though we are here armed with
better data.

The mean $\mean{X}$ is shown in Fig.~\ref{fig:CCmom1} and from fitted
2nd degree polynomials we estimate the limit $0.7753(2)$. The second
moment is again exactly $1$ but obviously we see some small
fluctuations. The behaviour is quite similar to that in
Fig.~\ref{fig:CGmom2} but with less noise. The third moment
$\mean{X^3}$ in Fig.~\ref{fig:CCmom3} prefers the limit $2.41(2)$, as
per fitted polynomials as before.  The fourth moment $\mean{X^4}$,
shown in Fig.~\ref{fig:CCmom4}, has a linear behavior just as for the
Gaussian case in Fig.~\ref{fig:CGmom4}. A rough estimate of its
behavior, based on $n\le 13$, is $\langle X^4 \rangle\sim 0.72(2) +
0.37(1) n$.

\begin{figure}
  \includegraphics[width=3.4in]{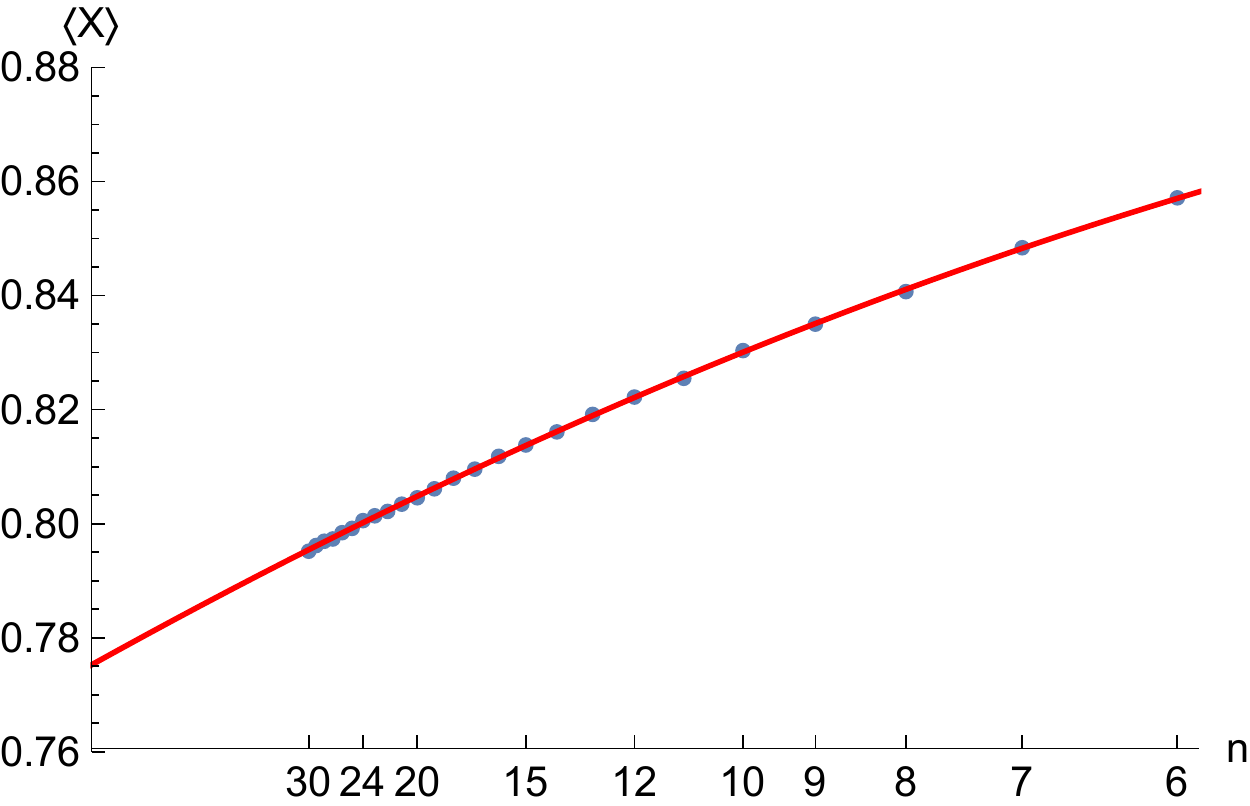}
  \caption{\label{fig:CCmom1} Complex circular entries. Mean value
    $\mean{X}$ versus $1/n$ for $n=6,7,\ldots,30$ and the fitted
    polynomial $y=0.7753 + 0.634x - 0.862x^2$ where $x=1/n$. Error
    bars are smaller than the points.}
\end{figure}

\begin{figure}
  \includegraphics[width=3.4in]{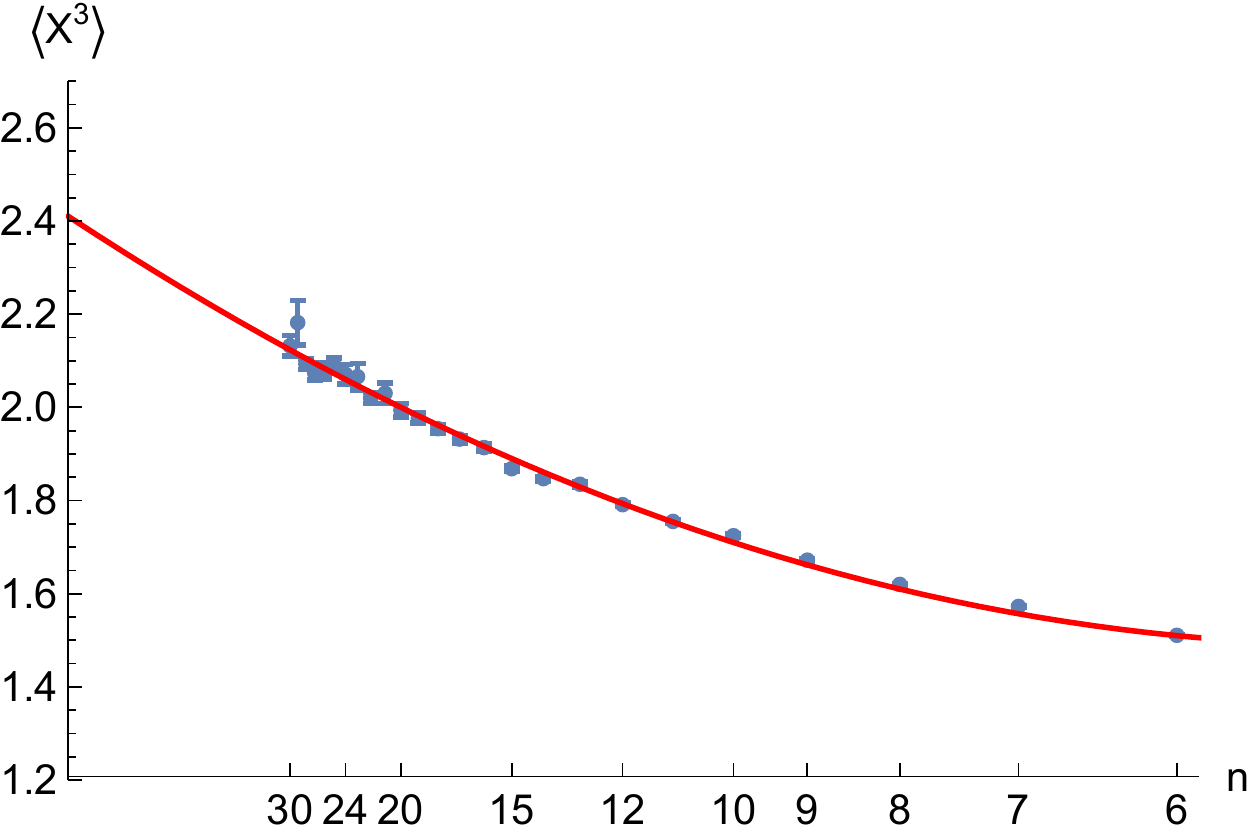}
  \caption{\label{fig:CCmom3} Complex circular entries. Mean value
    $\mean{X^3}$ versus $1/n$ for $n=6,7,\ldots,30$ and the fitted
    polynomial $y=2.41 - 9.4x + 24 x^2$.}
\end{figure}

\begin{figure}
  \includegraphics[width=3.4in]{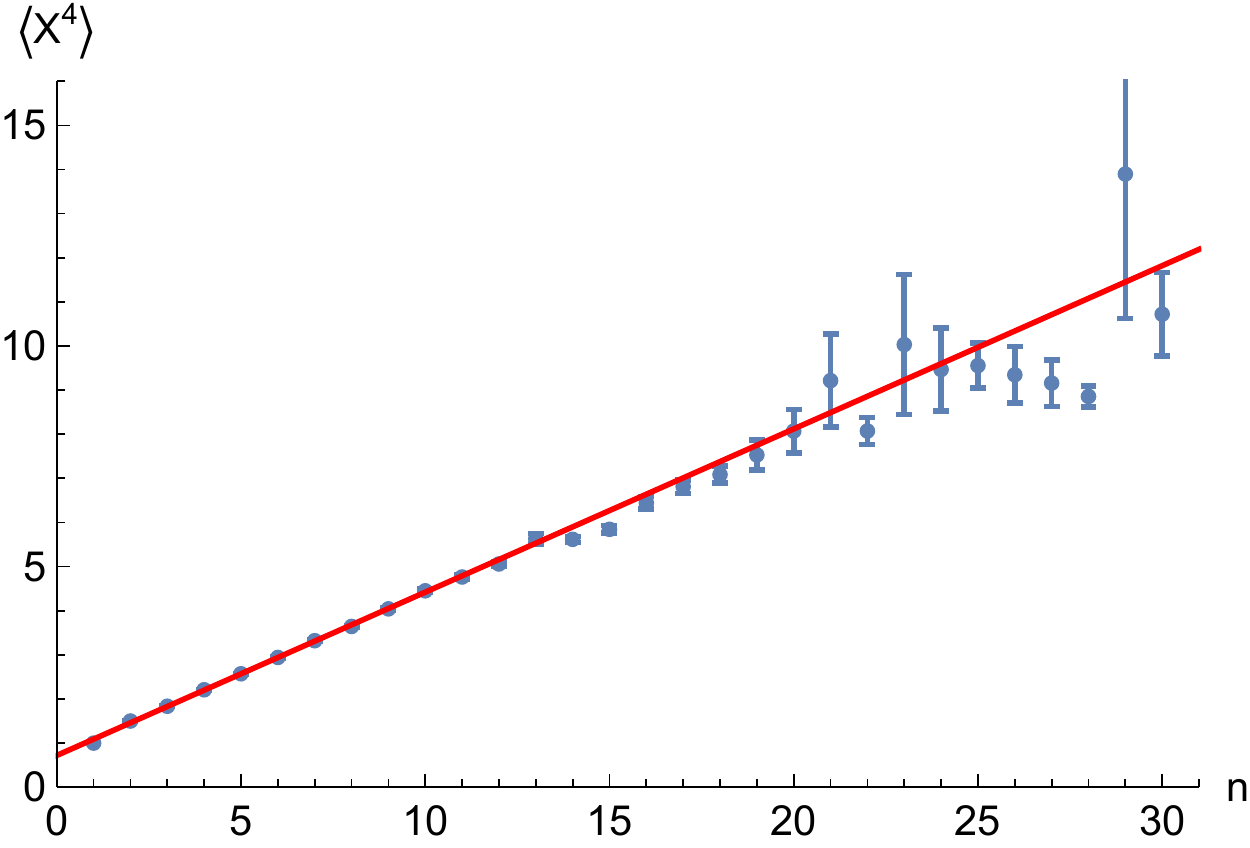}
  \caption{\label{fig:CCmom4} Complex circular entries. Mean value
    $\mean{X^4}$ versus $n$ for $n=1,2,\ldots,30$ and the line $y=0.72
    + 0.37 n$.}
\end{figure}

The distribution density $f(x)$ looks very similar to that of the
Gaussian case. To estimate the limit distribution function $F(x)$ we
use the ansatz $F(x,n)=C_0+C_1(x)/n+C_2(x)/n^2$ on the points $n\ge
5$.  Note that the linear ansatz used in the Gaussian case is not
sufficient here.  In Fig.~\ref{fig:CClogFx} we show a log-log plot of
the limit $F(x)$.  The red line has slope $2$ so again we expect
$F(x)\propto x^2$ for small $x$. The inset shows the ratio $F(x)/x^2$
and, despite the noise setting in at $x<0.01$, we estimate
$F(x)/x^2\sim 2.135(10)$. The error bar includes both errors from
excluding points in the $F(x,n)$ ansatz and errors depending on which
points $x$ to include (we have used $0.01\le x \le 0.08$). Clearly, as
shown by the inset figure, the rule $F(x)\sim 2.135 x^2$ breaks down
for $x\gtrsim 0.1$. Including the finite-size scaling we find the
following rule useful for $x\le 0.1$ for $n\ge 5$:
\begin{equation}\label{CCFasym}
  F(x,n) \sim \left(2.135 - \frac{8}{n} + \frac{15}{n^2}\right)\, x^2
\end{equation}

\begin{figure}
  \includegraphics[width=3.4in]{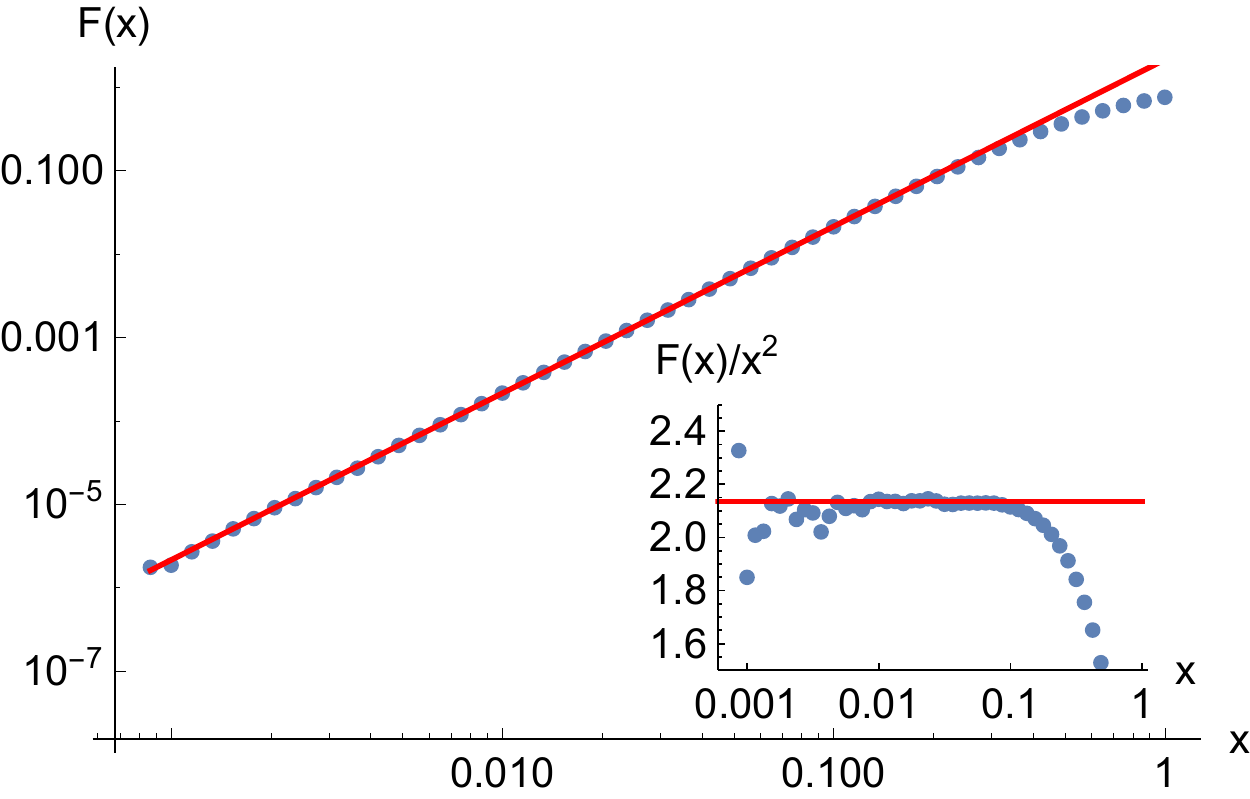}
  \caption{\label{fig:CClogFx} Complex circular entries. Log-log plot
    of limit $F(x)$ versus $x$ for $x\le 0.1$. The red line, having
    slope $2$, corresponds to the approximation $F(x)\sim 2.135
    x^2$. The inset shows $F(x)/x^2$ versus $\log x$ and the red line
    is the constant $y=2.135$. Noise becomes noticable for $x\lesssim
    0.01$.}
\end{figure}

Information on higher order terms are easier found when studying the
distribution of the squares. The distribution in Eq.~\eqref{CCFasym}
translates to $f_2(x,n) = 2.135 - 8/n + 15/n^2$ so that the limit
density is just the constant $2.135$ for, say, $x\lesssim 0.01$.
However, plotting $f_2(x,n)$ reveals that the density functions are
very close to linear for $x\lesssim 0.1$ and all $n\ge 5$.  Using the
simple ansatz $f_2(x,n) = C_0 + C_1 x/n$ and fitting on the interval
$0.005\le x\le 0.08$ we find that the slope scales as $C_1 = -6.0(2) +
46(3)/n - 110(20)/n^2$ where the error bars mainly indicate
sensitivity to which points are included. Note that the constant
coefficient must scale as $C_0 = 2.135 - 8/n + 15/n^2$, see
Eq.~\eqref{CCFasym}. To conclude, after defining the coefficients
\begin{eqnarray}
  C_0 &=&2.135 - \frac{8}{n} + \frac{15}{n^2}\\
  C_1 &=& -3 + \frac{23}{n} - \frac{55}{n^2}
\end{eqnarray}
we obtain the finite-size scaling rules
\begin{eqnarray}
  F(x,n) &\sim& C_0 \,x^2 + C_1\,x^4 \label{CCFasym2}\\
  f(x,n) &\sim& 2C_0\,x + 4C_1\,x^3  \label{CCfasym2}\\
  F_2(x,n) &\sim& C_0\,x + C_1\,x^2  \label{CCF2asym2}\\
  f_2(x,n) &\sim& C_0 + 2C_1\,x      \label{CCf2asym2}
\end{eqnarray}
and their respective limits become
\begin{eqnarray}
  F(x) &\sim& 2.135 x^2 -3 x^4  \label{CCFlimit}\\
  f(x) &\sim& 4.27 x -12 x^3    \label{CCflimit}\\
  F_2(x) &\sim& 2.135 x - 3 x^2 \label{CCF2limit}\\
  f_2(x) &\sim& 2.135 -6 x      \label{CCf2limit},
\end{eqnarray}

In Figs.~\ref{fig:CCfx} and \ref{fig:CCf2x} we compare the rules for
$f(x,n)$ and $f_2(x,n)$ to their measured counterparts. They fit very
well over a surprisingly wide interval. Translating the limit
$f_2(x)=2.135 - 6x$ to $F(x)$ we find $F(x) = 2.135x^2-3x^4$ which
adds a correction term to Eq.~\eqref{CCFasym}.  Note that in
Eq.~\eqref{CCFlimit} this correction term is of order $x^4$ whereas it
was of order $x^3$ in the Gaussian case.

\begin{figure}
  \includegraphics[width=3.4in]{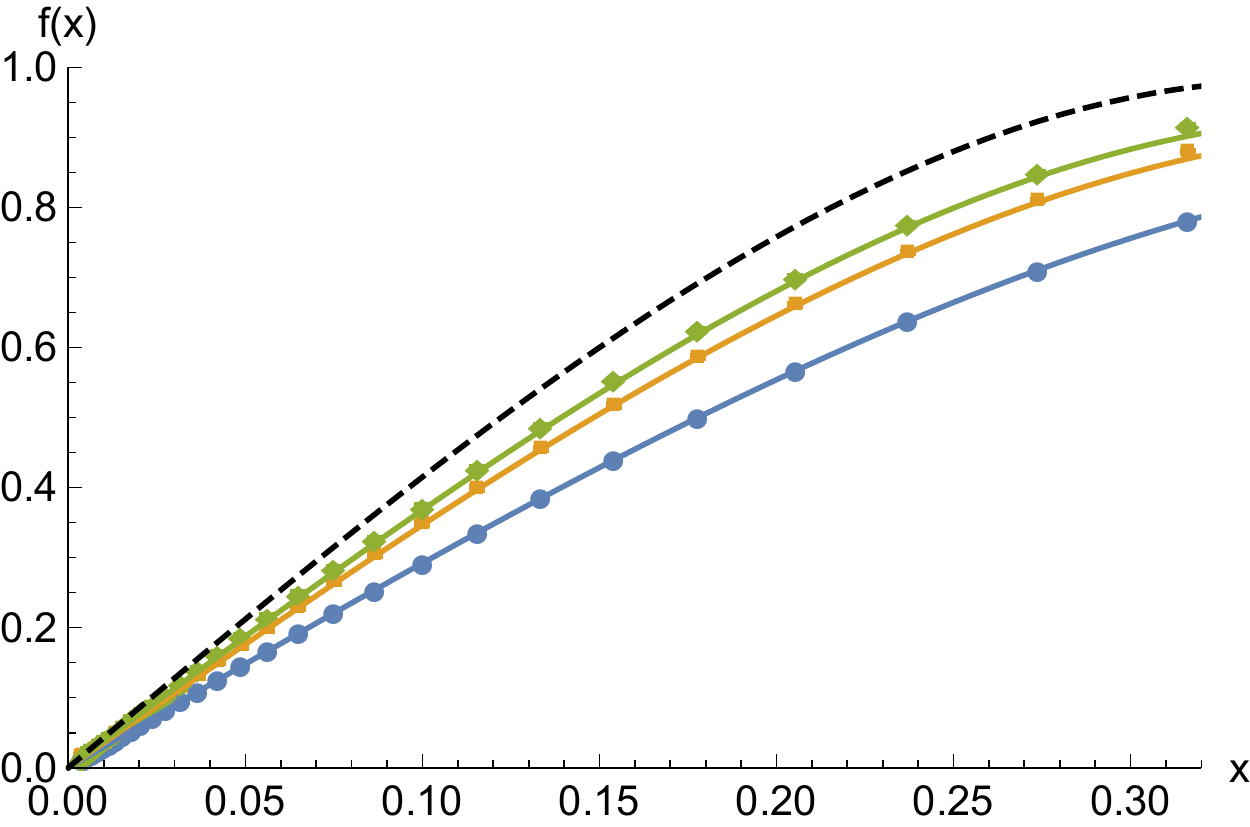}
  \caption{\label{fig:CCfx} Complex circular entries. Measured density
    $f(x,n)$ (points) and the scaling rule of Eq.~\eqref{CCfasym2} for
    $n=10$, $20$, $30$ (blue, orange, green curves; upwards) and
    $\infty$ (dashed black curve). Error bars are smaller than the
    points.}
\end{figure}

\begin{figure}
  \includegraphics[width=3.4in]{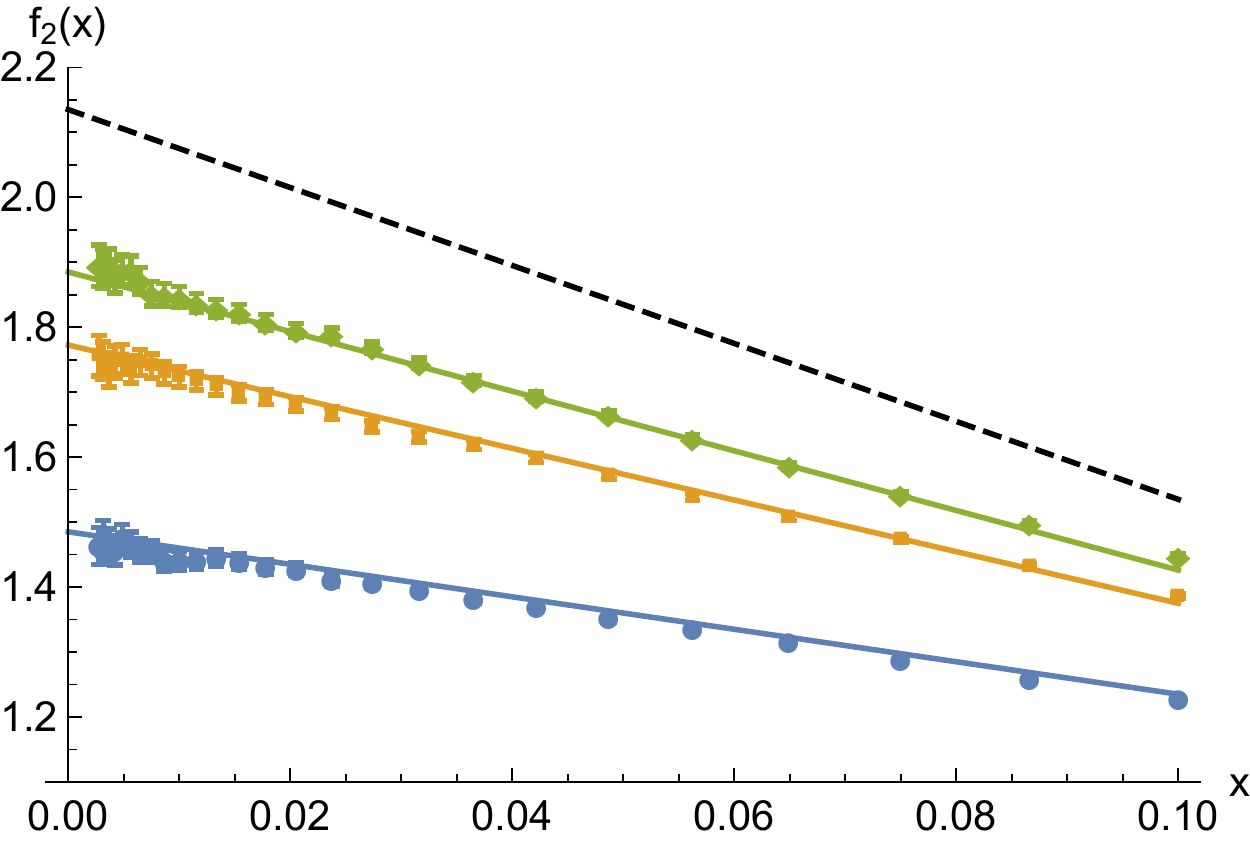}
  \caption{\label{fig:CCf2x} Complex circular entries. Measured density
    $f_2(x,n)$ (points) and the scaling rule of Eq.~\eqref{CCf2asym2} for
    $n=10$, $20$, $30$ (blue, orange, green curves; upwards) and
    $\infty$ (dashed black curve).}
\end{figure}


\section{Bernoulli-distributed entries}
We will here apply the approach in the previous section to a discrete
class of random matrices, that of Bernoulli-distributed $\pm 1$
entries where $\Pr(+1)=\Pr(-1)=1/2$. Matrices with
Bernoulli-distributed entries have been studied in the mathematics
literature, with an emphasis on the probability for small values of
$X$. In \cite{MR2483225} it was proven that with probability tending
to 1 $X$ is larger than $n^{(n/2-\epsilon)}$ for any fixed
$\epsilon>0$, and that it is likewise smaller than $n^{(n/2
  +\epsilon)}$ with probability tending to 1.  Those authors also
conjectured that the lower bound can be improved to $\exp(-c n)
n^{n/2}$ for some constant $c$, and we will comment more on this
later.

Our data consists of $10^6$ samples for
$n=1,2,\ldots,30$. Unfortunately we can not obtain quite the same
level of precision in our scaling analysis as for the complex Gaussian
and circular cases. Strong finite-size effects and erratic behavior
for smaller $n$ calls for larger matrices and many more samples.

Starting out with the first moment $\mean{X}$ in Fig.~\ref{fig:Bmom1}
we find that it is asymptotically $0.647(2)$. The second moment is of
course $1$ and the third moment is asymptotically $3.75(5)$ (not
shown). The fourth moment, as in the Gaussian case, appears to grow
linearly with $n$ but we only give the very rough estimate $\mean{X^4}
\approx 1.25(5) n - 1.0(5)$ due to its erratic behavior.

\begin{figure}
  \includegraphics[width=3.4in]{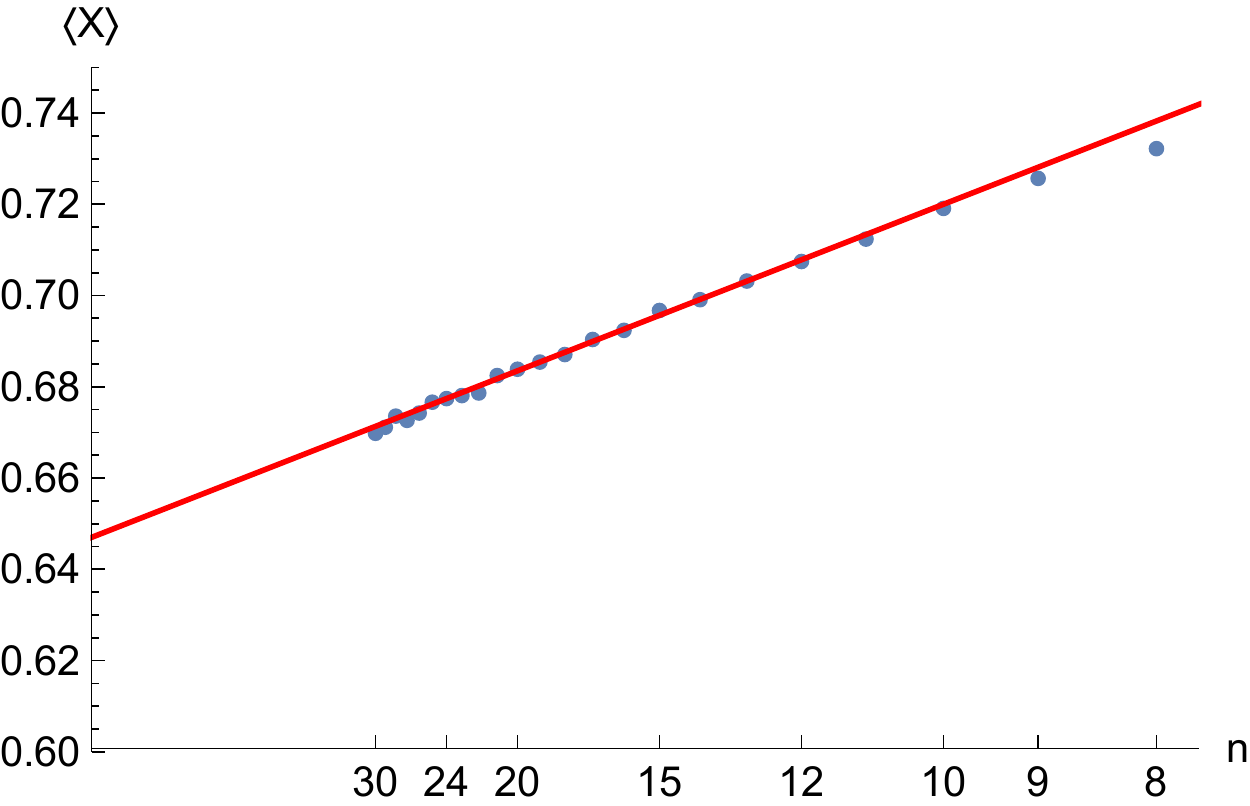}
  \caption{\label{fig:Bmom1} Bernoulli entries. Mean value $\mean{X}$
    versus $1/n$ for $n=8,9,\ldots,30$ and the line $y=0.647 + 0.73x$
    (red) where $x=1/n$. Error bars are smaller than the points.}
\end{figure}

In Fig.~\ref{fig:BlogFx} we show a log-log plot of $F(x)$ plotted
versus $\log x$ where $F(x)$ was obtained using a similar finite-size
scaling ansatz as in the previous cases. However, we note here that
care must be taken to only include matrices large enough since the
corrections-to-scaling are considerably larger in this case.  We have
only used $n\ge 15$ which contributes some noise to the estimated
$F(x)$ since we fit on fewer points. The red line in
Fig.~\ref{fig:BlogFx} has slope $1$ and corresponds to the estimate
$F(x)\sim 1.33(5) x$. The inset shows the ratio $F(x)/x$ which is
clearly approaching a limit around $1.33(5)$.

\begin{figure}
  \includegraphics[width=3.4in]{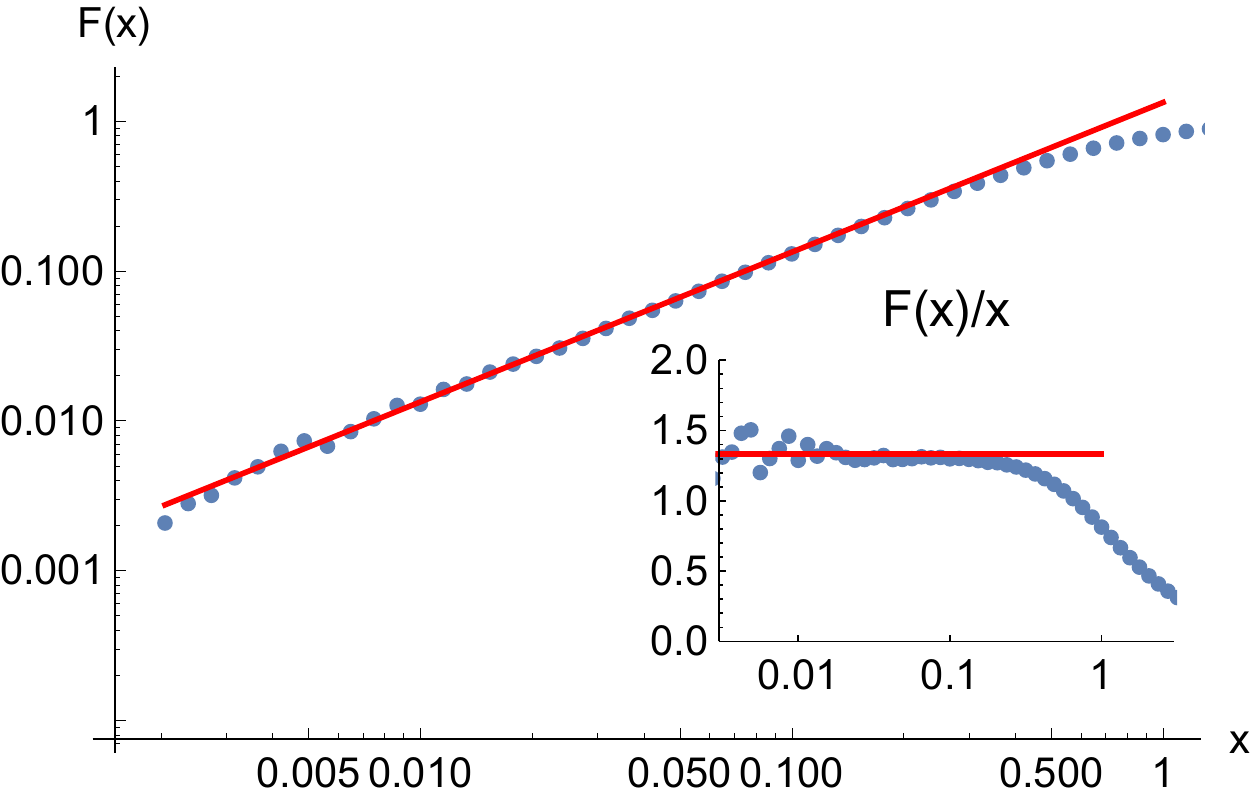}
  \caption{\label{fig:BlogFx} Bernoulli entries. Log-log plot of
    $F(x)$ and a line with slope $1$ (red) corresponding to the
    approximation $F(x)\sim 1.33 x$. The inset shows the ratio
    $F(x)/x$ versus $\log x$ together with the line $y=1.33$ (red).}
\end{figure}

Unfortunately our data are not good enough for a correction term of
higher order and pin-pointing the finite-size scaling would also be an
unreliable affair. We will simply translate our distribution function
into the distribution of $X^2$. In conclusion we thus find the
asymptotes
\begin{eqnarray}
  F(x) &\sim & 1.33 x \\
  f(x) &\sim & 1.33 \label{Bflimit} \\
  F_2(x) &\sim& 1.33 \sqrt{x}\\
  f_2(x) &\sim& 0.665/\sqrt{x} \label{Bf2limit}
\end{eqnarray}

Note here that we claim that $f_2(x)\to\infty$ when $x\to 0$ unlike
for the complex Gaussian and circular cases where $f_2(x)$ approached
a limit of $6.0(1)$ and $2.135(10)$ respectively.  In
Fig.~\ref{fig:Bdist12-n30} we show the measured $f(x,n)$ (inset) and
$f_2(x,n)$ for $n=30$ and compare them to the limits of
Eq.~\eqref{Bflimit} and \eqref{Bf2limit}.

\begin{figure}
  \includegraphics[width=3.4in]{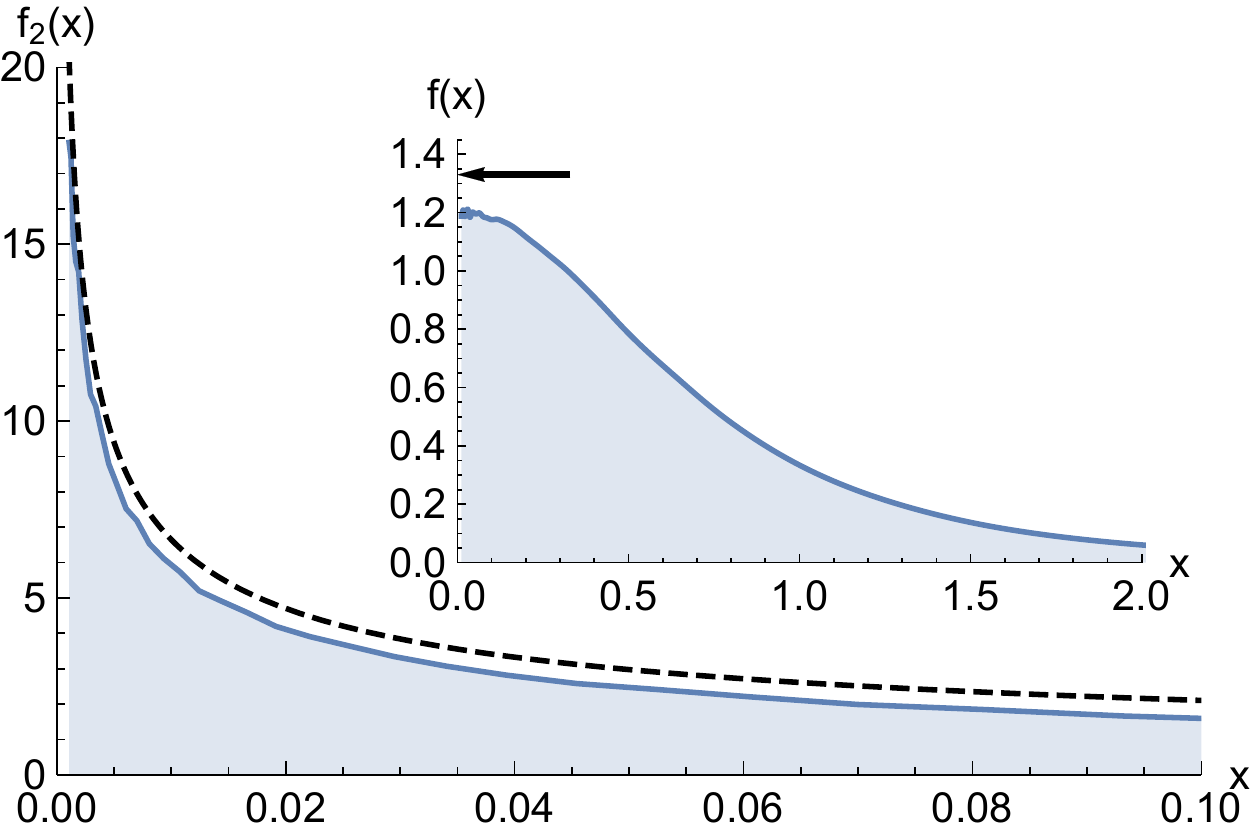}
  \caption{\label{fig:Bdist12-n30} Bernoulli entries. Measured density
    of $X^2$ for $n=30$, i.e.~$f_2(x,30)$, and the estimated limit $f_2(x) =
    0.665/\sqrt{x}$ (dashed black curve). The inset shows the measured
    density of $X$ for $n=30$, i.e.~$f(x,30)$. The black arrow indicates
    the limit $f(x)=1.33$.}
\end{figure}

Finally we note that our data is compatible with a strengthening of
the conjecture from \cite{MR2483225}
\begin{conjecture}
  $X$ is asymptotically almost surely larger than $$h(n)\exp(-n/2)
  n^{(n/2+1/4)},$$ where $h(n)$ is any function tending to 0 as
  $n\rightarrow \infty.$
\end{conjecture}


\section{Gaussian behaviour of minors of unitary matrices}
The computational hardness of Boson sampling as analysed in
\cite{aaronson:13} depends on the fact that certain submatrices of
Haar-random unitary matrices asymptotically behave like random
matrices with Gaussian entries.  Next we will investigate how close
the permanent of such a submatrix is to the permanent of a random
Gaussian matrix.

Let $\mathcal{U}(n)$ be the family of unitary $n\times n$-matrices.
We can generate random members of $\mathcal{U}(n)$ under the
Haar-measure in the following way~\cite{mezzadri:07}; produce a
complex random Gaussian matrix $A$ as above, find its QR-decomposition
with $R=(r_{i,j})$, let $\Lambda=(\lambda_{i,j})$ be the diagonal
matrix of normalized elements of $R$ so that
$\lambda_{i,j}=\delta_{i,j} r_{i,j}/|r_{i,j}|$ (where $\delta_{i,j}$
is the Kronecker delta), then $U=Q \Lambda$ is a random unitary matrix
from $\mathcal{U}$.

Now let $\mathcal{S}(m,n)$ be a family of random matrices obtained by
first generating a random unitary $m\times m$-matrix $U\in
\mathcal{U}(m)$, next let $U_n$ be the top-left $n\times n$-submatrix
of $U$ and set $S=\sqrt{m} U_n$. Then $S$ is a random matrix from the
family $\mathcal{S}(m,n)$.  

It is known that if $S\in\mathcal{S}(n^6,n)$ and 
$A\in\mathcal{G}(n)$ then the variation distance between them is
expected to be small, i.e., they have essentially the same probability
distribution.  The authors of Ref.~\cite{aaronson:13} prove a slightly
stronger result but they also think $n^6$ can be replaced by something
much smaller, say closer to $n^{2+\epsilon}$. We will use
distributions of the permanent to see if we can throw some light on
the problem.

We will use the Kolmogorov-Smirnov (KS) statistic $D =
\sup_x|F(x)-G(x)|$ as a measure of the distance between two empirical
distribution functions $F(x)$ and $G(x)$. For a two-sample test, we
reject the null hypothesis that they are the same (at significance
level $\alpha$) if $D>D_{\alpha}$ for certain $D_{\alpha}$. For
$\alpha=0.05$ and using $10^5$ samples for both distributions we get
$D_{\alpha} = 0.00607$.  We then first generate $S\in\mathcal{S}(n^a,
n)$ and compute $|\per(S)|/\sqrt{n!}$ for $10^5$ different $S$ and
then compare this distribution to that of $|\per(A)|/\sqrt{n!}$ for
$10^5$ complex Gaussian matrices $A$. We will see if the distance $D$
has an increasing or decreasing trend for different values of
$a$. Note that when $n^a$ is not an integer we just round to the
nearest integer. We use Mathematica's built-in routine for computing
the test statistic when comparing two distributions in a
Kolmogorov-Smirnov test as the value of $D$.

We have run this test for $5$ different $a$ and a wide range of $n$
for each $a$: $1\le n\le 34$ for $a=2$, $1\le n\le 32$ for $a=2.25$,
$1\le n\le 30$ for $a=2.50$, $1\le n \le 22$ for $a=2.75$ and $1\le
n\le 17$ for $a=3$. In Fig.~\ref{fig:ksvn} we show the KS-statistic
$D$ versus $n\ge 3$ for the different $a$.


For $a=2$ the values of $D$ are clearly increasing at first but there
is no clear trend beginning at $n\approx 28$, with $D$ staying at
roughly $0.08$. For $a=2.25$ there is a very weak increasing trend in
$D$. Excluding individual points from the line fit is not enough to
get a decreasing trend though. For $a=2.5$ there is a distinctly
decreasing trend but it would take $n\approx 75$ to pass a KS-test at
the $5$\% level. For $a=2.75$ the distributions actually pass a
KS-test for $n=19$, $21$, $22$ and for $a=3$ they pass it for $n=12$
and $13$.

Reading the trends in the KS-statistic $D$, it would thus appear that
matrices from $\mathcal{S}(n^a, n)$ are essentially indistinguishable
from complex Gaussian $n\times n$-matrices, in terms of their
permanents, when $a > 2.25$, while for $a<2.25$ they are not, and the
case of $a=2.25$ appears to be a separator between the two cases which
we cannot classify.  If this is a correct classification, rather than
an effect of slow convergence in terms of $n$ for the lower values of
$a$, then it would contradict the conjecture in
Ref.~\cite{aaronson:13} that $a\ge 2+\epsilon$ is enough.

It is possible that the curve for $a=2$ will begin to decrease for
larger $n$, in accordance with the conjecture from
Ref.~\cite{aaronson:13}. Nontheless, for $a<2.25$, we see a behaviour
which is distinct from the Gaussian case for the range of $n$ used
here. So, care must be taken in the analysis of Boson sampling
experiments where an effective value of $a$ close to, or equal to, 2
has been chosen if the number of Bosons is small.

\begin{figure}
  \includegraphics[width=3.4in]{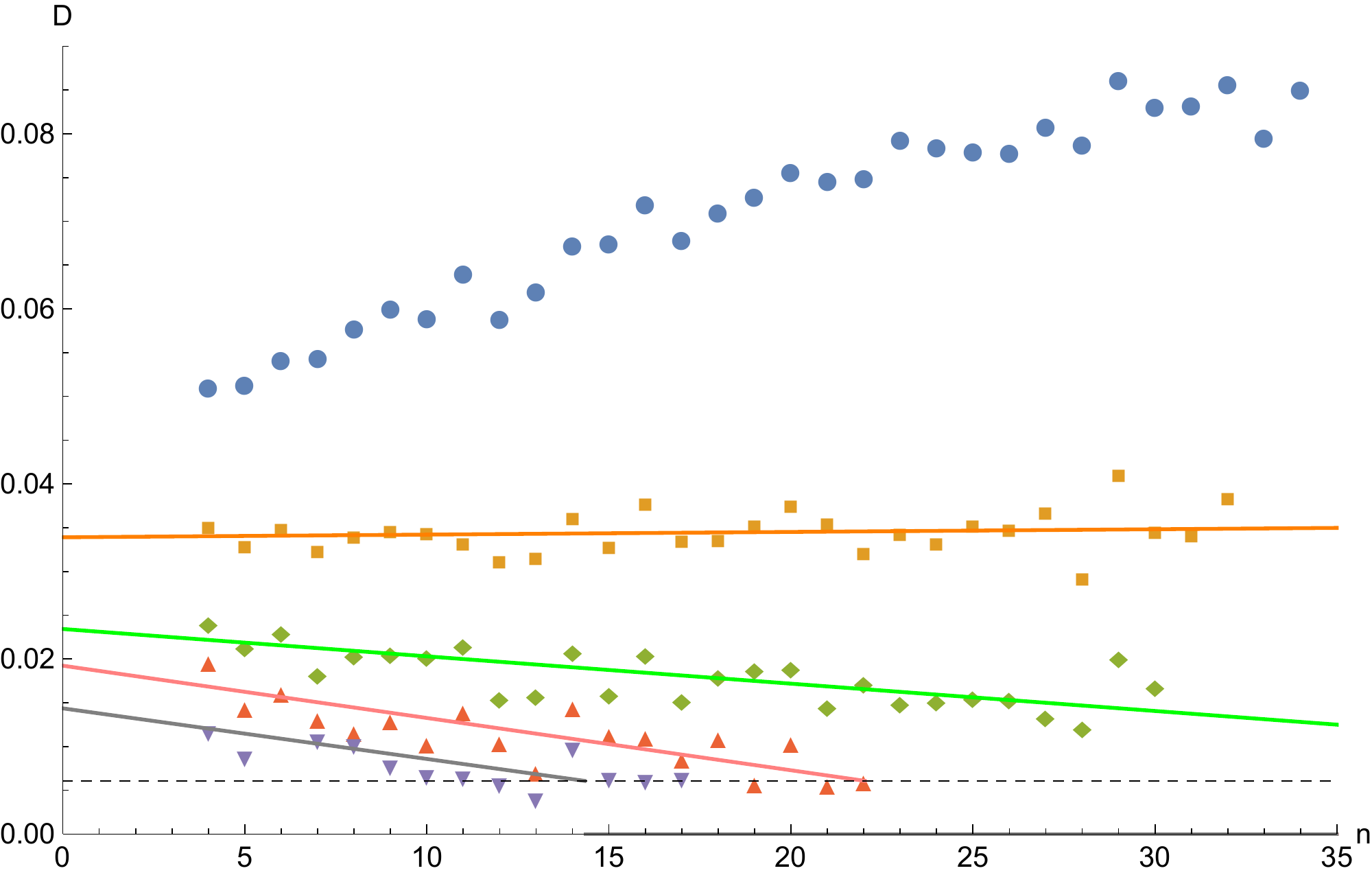}
  \caption{\label{fig:ksvn} KS-statistic $D$ plotted versus $n$ for
    the five different $a$. Downwards in figure are data sets for
    $a=2.00$ (blue points), $2.25$ (orange squares), $2.50$ (green
    diamonds), $2.75$ (pink up-triangles) and $3.00$ (purple
    down-triangles). Both distributions are based on $10^5$ samples so
    $D_{0.05}=0.00607$ (dashed line). The fitted lines are
    $0.034+0.000031x$ for $a=2.25$, $0.023-0.00031x$ for $a=2.50$,
    $0.019-0.00060x$ for $a=2.75$ and $0.014-0.00058x$ for $a=3$.}
\end{figure}


\section{Conclusions}

In order to compute output probabilities for Boson sampling
experiments \cite{aaronson:13} one has to compute the permanent of the
associated unitary matrices. Here we have presented a software package
for doing such computations efficiently, both on serial and parallel
machines. Our programs are efficient enough to allow us to beat the
previous world record for computation of permanents in a substantial
way, despite the fact that the previous record was set on a far larger
cluster \cite{Tian}.

Our package also has specialised functions for matrices of limited
bandwidth, running in time $\mathcal{O}(2^k n^2)$ for matrices of
bandwidth $k$, and in linear time for fixed $x$.  This makes it
possible to classically simulate a Boson sampling system of depth
$\mathcal{O}(\log{n})$ in polynomial time

We have used our software package to perform a large scale simulation
study of the anti-concentration conjecture for permanents
\cite{aaronson:13}. Here we find that the conjecture agrees well with
the conjecture, both for complex Gaussian matrices and other matrix
classes.  We also investigated how well the permanent of a minor of
size $n^a$ of an $n\times n$ Haar-random unitary matrix can be
approximated by the permanent of a random Gaussian matrix.  Here we
find some possible tension with the most optimistic version of a
conjecture from Ref.~\cite{aaronson:13}.


\begin{acknowledgments}
  The simulations were performed on resources provided by the Swedish
  National Infrastructure for Computing (SNIC) at High Performance
  Computing Center North (HPC2N).  The second author was supported by
  The Swedish Research Council grant 2014--4897.
\end{acknowledgments}

\appendix


\section{A parallel version of Ryser's algorithm}\label{alg}

We here collect the algorithms necessary for computing the permanent of general matrices 
on a parallel computer. Fortran and Mathematica implementations can be
freely downloaded and used from our website~\cite{perm-code}.

Distributing the computation equally on a number of nodes is of course
easily done.  The computation is a sum with $2^{n-1}$ terms and we
only need to split the sum into equal parts (or as equal as is
possible).  To distribute a sequence $(r_0, r_1, \ldots, r_{n-1})$ as
evenly as possible into $m$ subsequences $(s_0, s_1, \ldots,
s_{m-1})$, where $s_i = (r_k, r_{k+1}, \ldots, r_{k+\ell-1})$, we need
to find $k$ and $\ell$ for each $i=0,\ldots,m-1$:
\begin{itemize}
  \setlength\itemsep{0ex}
\item distribute$(n, m, i, k, \ell)$
\item In: $n\ge 0$, $m\ge 1$, $0\le i < m$
\item Out: $k$, $\ell$
\item[1] $q := \lfloor n/m \rfloor$
\item[2] $r := n \bmod m$
\item[3] $k := i\cdot q + \min(i,r)$
\item[3] $\ell := q + \mathcal{I}(i<r)$
\end{itemize}
Here $\mathcal{I}(s)$ is an indicator function returning $1$ if
statement $s$ is true and $0$ if $s$ is false.

The sequence in question is of course the Gray-code sequence.  To find
the first code in a subsequence we need to unrank the Gray-code,
i.e.~compute the $k$th code $x=(x_1,x_2,\ldots)$ in the Gray-code
sequence. It is common practice that this is obtained as
\begin{equation}
  x := \mathrm{xor}(k, \mathrm{rshift}(k))
\end{equation}
where $\mathrm{xor}$ is the bit-wise exclusive-or function of two
integers and $\mathrm{rshift}$ denotes a bit-wise shift of an integer
one step to the right.

Computing the next Gray-code in the sequence is also common knowledge,
see e.g.~Ref.~\cite{wilf:78}, but we include it here for completeness:
\begin{itemize}
  \setlength\itemsep{0ex}
\item nextset$(t, j, x)$
\item In: $t=\pm 1$ and binary vector $x$.
\item Out: integer $j$ and updated $t$ and $x$.
\item[1] $j := 1$  (first position of $x$)
\item[2] $t := -t$
\item[2] if $t = 1$ then
\item[3] \quad while $x_j = 0$ do
\item[4] \qquad $j := j+1$
\item[5] \quad end do
\item[6] \quad $j = j + 1$
\item[6] end if
\item[7] $x_j:=1-x_j$
\end{itemize}

The permanent of a $n\times n$-matrix is a sum of $2^{n-1}$ terms
which we want to distribute over, say, $m$ nodes (or threads).  Each
node then computes the partial sum $S_k$ for $k=0,1,\ldots,m-1$. For
details, see Ref.~\cite{wilf:78}.

\begin{itemize}
  \setlength\itemsep{0ex}
\item subpermanent$(A, n, m, k, S)$
\item In: $n\times n$-matrix $A=(a_{ij})$, integers $0\le k < m$.
\item Out: partial sum $S$.
\item[1] $\mathrm{distribute}(2^{n-1}, m, k, r, \ell)$
\item[2] $x := \mathrm{xor}(r, \mathrm{rshift}(r))$ (where $x=(x_1,x_2,\ldots,x_n)$)
\item[3] $t := (-1)^r$
\item[4] $S := 0$
\item[5] for $i=1,2,\ldots,n$ let $w_i := a_{i,n} - \tfrac{1}{2}\sum_{j=1}^n a_{i,j}$
\item[6] for $j=1,2,\ldots,n-1$ where $x_j=1$ do
\item[7] \quad for $i=1,2,\ldots,n$ let $w_i := w_i + a_{i,j}$
\item[8] end do
\item[9] do $\ell$ times
\item[10] \quad $p := \prod_{i=1}^n w_i$
\item[11] \quad $\mathrm{nextset}(t,j,x)$
\item[12] \quad $S:=S + t\cdot p$
\item[13] \quad $z := 2x_j-1$
\item[14] \quad for $i=1,2,\ldots,n$ let $w_i := w_i + z\cdot a_{i,j}$
\item[15] end do
\end{itemize}

Collecting and adding up the partial sums is easy.
\begin{itemize}
  \setlength\itemsep{0ex}
\item permanent$(A,n,m,S)$
\item In: $n\times n$-matrix $A$ and integer $m\ge 1$.
\item Out: permanent $S$.
\item[1] for $k=0,1,\ldots m-1$ do
\item[2] \quad $\mathrm{subpermanent}(A,n,m,k,S_k)$
\item[3] end do
\item[4] $S := 2\,(-1)^n\,\sum_{k=0}^{m-1} S_k$
\end{itemize}

For completeness we also describe Kahan summation~\cite{kahan:65}.
Consider the following standard summation loop computing
$S:=a_1+\cdots+a_n$,
\begin{itemize}
  \setlength\itemsep{0ex}
  \item[1] $S := 0$
  \item[2] for $i=1,2,\ldots,n$ do
  \item[3] \quad $S := S + a_i$
  \item[4] end do
\end{itemize}
In Kahan summation we do instead the following
\begin{itemize}
  \setlength\itemsep{0ex}
\item[1] $S := 0$
\item[2] $b := 0$
\item[3] for $i=1,2,\ldots,n$ do
\item[4] \quad $c := a_i - b$
\item[5] \quad $t := S + c$
\item[6] \quad $b := (t - S) - c$
\item[7] \quad $S := t$
\item[8] end do
\end{itemize}


\section{Algorithm for Sparse Matrices}\label{alg2}
The main observation behind our improvement for sparse matrices on
Ryser's method comes from the observation that in a sparse matrix many
of the products in Ryser's formula \eqref{ryser} will be zero. By
avoiding sets $J$ which are guaranteed to lead to a zero sum in the
innermost product we may achieve a speed-up.

If we interpret the matrix $A$ as the adjacency matrix of an
edge-weighted graph $G$, where vertices may have loops, we find that a
set $J$ can only lead to a non-zero product if $J$ is a dominating set
in $G$, i.e., every vertex in $G$ has at least one neighbor in $J$.
Finding all dominating sets can be done in exponential time, with a
basis smaller than 2 \cite{Fomin}, but we can instead use a faster
approximate algorithm which still leads to a speed-up over the general
version of Ryser's formula.

A subset $S$ of the vertex set of $G$ is \emph{domination restricting}
if every dominating set of $G$ must contain a vertex from $S$.  The
full vertex set of $G$ is domination restricting, as is the
neighourhood of a single vertex.  We say that a list of sets
$L=(S_1,S_2,\ldots,S_t)$ is domination restricting if each set $S_i$
is domination restricting and the sets are pairwise disjoint.  We now
note that every dominating set $J$ in $G$ must have a non-empty
intersection with each set in $L$.  So, if we use all sets $J$ with
this property we will include all dominating sets $J$, and some sets
which may not be dominating, while excluding a potentially large
number of sets.  Below we give a randomized greedy algorithm for
constructing a useful list $L$.

We say that a matrix $A$ is $d$-sparse if every row and column
contains at most $d$ non-zero entries.  For this type of matrix a good
choice of $L$ will lead to an exponential speed-up over the basic
version of Ryser's formula. Let us now re-state and prove
the theorem in Sec.~\ref{sec:iassm}.

\begin{theorem}
	Let $A$ be a $d$-sparse $n \times n$ matrix. Then the permanent of
    $A$ can be computed in time
	\[\mathcal{O}(n2^n (1-2^{-d})^{n/d^2})\]
\end{theorem}
\begin{proof}
  Let $G^2$ be the square of the graph $G$ associated with $A$,
  i.e. the graph where two vertices are adjacent if they are at
  distance at most 2 in $G$. The graph $G^2$ has degree at most $d^2$,
  so if $n>d^2$ we can properly colour the vertices of $G^2$ using
  $d^2$ colours.

  Now we can construct a domination restricting list $L$ by taking a
  colour class of size at least $n/d^2$ from $G^2$ and for every
  vertex in that colour class including its neighbourhood as a set in
  $L$.  This gives us a list $L$ with $n/d^2$ sets, each of size $d$.

  We will now use all sets $J$ constructed by taking a non-empty
  subset of each set $S_i$ in $L$ and an arbitrary subset of the
  vertices not in $L$.  The number of such sets is
  $(2^d-1)^{n/d^2}2^{n-n/d}$.
\end{proof} 
The degree bounds in the theorem are exact for graphs which do not
contain short cycles and when such cycles are present we will
typically see a larger speed-up.  For non-symmetric $A$ we may also
gain more by instead taking $G$ to be a directed graph, where a
dominating set now means that each vertex has an out-neighbour in the
set.

\begin{itemize}
  \setlength\itemsep{0ex}
\item sparsepermanent$(A, D, p)$
\item In: sparse $n\times n$-matrix $A=(a_{ij})$ without $0$-rows or
  $0$-columns and greedy partition $D$ of $\{1,\ldots,n\}$ (see
  below).
\item Out: permanent $p$.
\item[1] Assume $D:=\{S_1,\ldots,S_d, T\}$
\item[2] $p:=0$
\item[3] for all $s_{\ell}\subseteq S_\ell$, $s_\ell\neq\emptyset$, $\ell=1,\ldots,d$
\item[4] \quad for all $t\subseteq T$
\item[5] \qquad $J=s_1\cup\ldots\cup s_d\cup t$
\item[6] \qquad $p:=p + (-1)^{|J|}\prod_{i=1}^n\sum_{j\in J} a_{i,j}$
\item[7] \quad end do
\item[8] end do
\item[9] $p:=p\cdot (-1)^n$
\end{itemize}

\begin{itemize}
  \setlength\itemsep{0ex}
\item greedypartition$(A,D)$
\item In: sparse $n\times n$-matrix $A=(a_{ij})$ without $0$-rows or $0$-columns.
\item Out: partition $D=\{S_1,\ldots, S_d, T\}$ of $\{1,\ldots,n\}$.
\item[ 1] for $i,j=1,\ldots,n$ let $b_{i,j}:=1$ if $a_{i,j}\neq 0$, otherwise $b_{i,j}=0$.
\item[ 2] for $i=1,\ldots,n$ let $\delta_i=\sum_{j=1}^n b_{i,j}$ (out-degree of $i$)
\item[ 3] let $d:=0$ and $V:=\{1,\ldots,\}$
\item[ 4] while $V\neq\emptyset$ do
\item[ 5] \quad $k:=\arg\min\{\delta_i: i\in V\}$ ($k$ has min degree)
\item[ 6] \quad $N_k:=\{\ell: b_{k,\ell}=1\}$ (neighbours of $k$)
\item[ 7] \quad $d:=d+1$
\item[ 8] \quad $S_d:=N_k$
\item[ 9] \quad $V:=V\setminus\{k\}$
\item[10] \quad for all $\ell\in V$ where $N_k\cap N_\ell\neq\emptyset$ do
\item[11] \qquad $V:=V\setminus \{\ell\}$
\item[12] \quad end do
\item[13] end do
\item[14] $T:=\{1,\ldots,n\}\setminus (S_1\cup\ldots\cup S_d)$
\item[15] $D:=\{S_1,\ldots, S_d, T\}$
\end{itemize}

Note that it is often beneficial to choose the minimum element of step
(5) at random.  Then run the partition algorithm several times and
pick the result which minimises the number
\begin{equation}
  2^{|T|} \prod_{\ell=1}^d (2^{|S_\ell|} - 1),
\end{equation}
which is the total number of sets enumerated in the sparsepermanent
algorithm above.


\section{Algorithm for  Matrices of limited Bandwidth }\label{alg3}
Here we describe our algorithm for computing the permanent of matrices
with bounded bandwidth.

\begin{itemize}
  \setlength\itemsep{0ex}
\item bandpermanent$(A, k, p)$
\item In: $n\times n$-matrix $A=(a_{ij})$,  bandwidth $0\le k \leq n$.
\item Out: permanent $p$.
\item[1] $D:=\mathrm{diag}(x_1, x_2, \ldots, x_n)$ (where the $x_i$
  are formal variables)
\item[2] $C:= (AD)\bar{1}_n$
\item[3] $p:=1$
\item[4] for $i=1,2,\ldots,n$ do
\item[5] \quad $p := p\cdot C_i$
\item[6] \quad In $p$, set $x_{i - k - 1} = 1$ and  set $x_j^2 =0$ for all $j$ 
\item[7] end do
\item[8] In $p$, set $x_i =1$ for all $i$ 
\end{itemize}
Note that step 2 should be done with matrix sparsity in mind to avoid
a quadratic overhead computational cost.


\end{document}